%% file: main.tex
\documentclass[11pt]{article}
\input{preamble}

\title{QSETH strikes again: finer quantum lower bounds for lattice problem, strong simulation, hitting set problem, and more}

\author{Yanlin Chen\thanks{
University of Maryland (QuICS)	 {\tt yanlin@umd.edu}. Most of this work was completed while the author was at Centrum Wiskunde $\&$ Informatica (QuSoft). }
\and 
Yilei Chen\thanks{Tsinghua University {\tt chenyilei.ra@gmail.com}. Supported by Tsinghua University startup funding, and Shanghai Qi Zhi Institute Innovation Program SQZ202405.}
\and 
Rajendra Kumar\thanks{Indian Institute of Technology Delhi {\tt rajendra@cse.iitd.ac.in}. Supported by the Chandrukha New Faculty Fellowship from IIT Delhi and the Prime Minister Early Career Research Grant from the Anusandhan National Research Foundation (ANRF).}
\and 
Subhasree Patro\thanks{Eindhoven University of Technology {\tt patrofied@gmail.com}}
\and 
Florian Speelman\thanks{University of Amsterdam and QuSoft {\tt f.speelman@uva.nl}. Supported by the Dutch Ministry of Economic Affairs and Climate Policy (EZK), as part of the Quantum Delta NL program, and the project Divide and Quantum `D\&Q' NWA.1389.20.241 of the program `NWA-ORC', which is partly funded by the Dutch Research Council (NWO).}}

\date{}

\begin{document}
\pagenumbering{roman}

\maketitle
\begin{abstract}
Despite the wide range of problems for which quantum computers offer a computational advantage over their classical counterparts, there are also many problems for which the best known quantum algorithm provides a speedup that is only quadratic, or even subquadratic.
Such a situation could also be desirable if we \emph{don't want} quantum computers to solve certain problems fast - say problems relevant to post-quantum cryptography. 
When searching for algorithms and when analyzing the security of cryptographic schemes, we would like to have evidence that these problems are difficult to solve on quantum computers; \emph{but how do we assess the exact complexity of these problems?} 

For most problems, there are no known ways to directly prove time lower bounds, however it can still be possible to relate the hardness of disparate problems to show \emph{conditional} lower bounds. This approach has been popular in the classical community, and is being actively developed for the quantum case \cite{BasicQSETH-Aaronson-2020, BPS21, BLPS22, ABLPS22}.

In this paper, by the use of the QSETH framework \cite{BPS21} we are able to understand the quantum complexity of a few natural variants of CNFSAT, such as parity-CNFSAT or counting-CNFSAT, and also are able to comment on the non-trivial complexity of approximate versions of counting-CNFSAT.
Without considering such variants, the best quantum lower bounds will always be quadratically lower than the equivalent classical bounds, because of Grover's algorithm; however, we are able to show that quantum algorithms will likely not attain even a quadratic speedup for many problems.
These results have implications for the complexity of (variations of) lattice problems, the strong simulation and hitting set problems, and more.
In the process, we explore the QSETH framework in greater detail and present a useful guide on how to effectively use the QSETH framework.
\end{abstract}

\thispagestyle{empty}
\newpage
\setcounter{tocdepth}{2}
{
  \hypersetup{linkcolor=black}
  \tableofcontents
}
\newpage
\pagenumbering{arabic}
\section{Introduction}

A popular classical hardness conjecture known as the Strong Exponential Time-Hypothesis (SETH) says that determining whether an input CNF formula is satisfiable or not cannot be done in $\mathcal{O}(2^{n(1-\delta)})$ time for any constant $\delta>0$ \cite{IP01,IPZ01}. Several fine-grained lower bounds based on SETH have been shown since then; see \cite{Vas15, Vas19} for a summary of many such results. Some of the SETH-hard problems are building blocks for fine-grained cryptography~\cite{DBLP:conf/stoc/BallRSV17,DBLP:conf/crypto/LaVigneLW19}.
Besides finding a satisfying assignment, natural variants of the \CNFSAT{} problem include computing the count or the parity of the count of satisfying assignments to a \CNF{} formula -- $\#$SETH and $\oplus$SETH conjecture complexity of these problems, respectively. These conjectures are weaker (i.e., more believable) than SETH, and can still be used to show fine-grained hardness of various problems~\cite{CDLMNOPSW16}.

When considering quantum computation, the SETH conjecture is no longer true, as using Grover's algorithm for unstructured search~\cite{grover1996fast} one can solve the \CNFSAT{} problem in $2^{\frac{n}{2}}\cdot \poly(n)$ time. Aaronson, Chia, Lin, Wang, and Zhang assume this Grover-like quadratic speedup is nearly optimal for CNFSAT and (independent of \cite{BPS21}) initiate the study of quantum fine-grained complexity~\cite{BasicQSETH-Aaronson-2020}. However, conjectures such as $\#$SETH or $\oplus$SETH are likely to still hold in the quantum setting because a Grover-like quantum speedup is not witnessed when the task is to compute the total number of satisfying assignments or the parity of this number. This situation can in some cases be exploited to give better quantum lower bounds than one would get from the conjectured quantum lower bound for the vanilla \CNFSAT{} problem. This makes it at least as relevant (if not more) to study variants of \CNFSAT{} and their implications in the quantum setting, as has been done classically. In fact, motivated by this exact observation, Buhrman, Patro, and Speelman~\cite{BPS21} introduced a framework of Quantum Strong Exponential-Time Hypotheses (QSETH) as quantum analogues to SETH, with a striking feature that allows one to `technically' unify conjectures such as quantum analogues of $\oplus$SETH, $\#$SETH, maj-SETH, etc.\ under one \emph{umbrella} conjecture. 


\paragraph{The QSETH framework}

In their framework, Buhrman et al.\ consider the problem in which one is given a formula or a circuit representation of a Boolean function $f:\{0,1\}^n \rightarrow \{0,1\}$ and is asked whether a property $\propertyP\coloneqq\big( \propertyP_n\big)_{n\in\mathbb{N}}$ where $\propertyP_n: \mathcal{D} \subseteq \{0,1\}^{2^n}\rightarrow \{0,1\}$ on the truth table\footnote{Truth table of a formula $\phi$ on $n$ variables, denoted by $tt(\phi)$, is a $2^n$ bit string derived in the following way $tt(\phi)=\bigcirc_{a \in \{0,1\}^n}\phi(a)$; the symbol $\circ$ denotes concatenation.} of this formula evaluates to $1$. They conjectured that when the circuit representation is obfuscated enough then for \textit{most} properties $\propertyP$ (that are compression-oblivious properties as we will see in \cref{def:CompressionOblivious}), the time taken to compute $\propertyP_n$ on the truth table of $\poly(n)$-sized circuits is lower bounded by $\Q(\propertyP_n)$, i.e.\ the $1/3$-bounded error quantum query complexity of $\propertyP_n$, on all bit strings of length $2^n$. 

It is not hard to see that such a conjecture cannot be true for \emph{all} properties. In principle, one can construct properties for which the above statement would not hold. For instance, consider a property $\propertyP$ that is trivial on truth tables of small formulas (i.e., $\poly(n)$ size) but complicated on formulas of longer length. These kinds of properties are likely to have very high quantum query complexity, but in reality, it will be trivial to compute such a $\propertyP_n$ of $\propertyP$ on formulas of $\poly(n)$ size. In order to prevent such scenarios the authors in \cite{BPS21} introduce the notion of compression-oblivious properties which they believe encompasses most of the naturally occurring properties. See Sections~2.2 and 2.3 of \cite{BPS21} for a detailed discussion on this topic and also see \cite{CCKPS25Qnaturalproof} for some new observations about the notion of compression-oblivious properties. To give a bit of intuition, first consider the set of truth tables corresponding to the set of $\poly(n)$ size formulas on $n$ variables and consider the set of all the $2^n$ bit strings. Compression-oblivious properties are those properties for which one cannot save in computational time to compute them on a string from the former set in comparison to computing the same property on strings from the latter set. More formally,\footnote{The definition of compression-oblivious properties as stated in this paper is a more formal version of its original definition in~\cite{BPS21}. Also, see \cite{CCKPS25Qnaturalproof} for a discussion on this topic.} 

\begin{defn}[$\ACp$- and $\AC$-Compression-Oblivious Properties \cite{BPS21,CCKPS25Qnaturalproof}]
\label{def:CompressionOblivious}
Let $p \in \mathbb{N}$. We say a property $\propertyP$ is \emph{$\ACp$-compression-oblivious}\footnote{
We say a language $L \in \ACarg{2}{p}$ iff there exists a family of Boolean circuits $\{C_n\}_{n \in \mathbb{N}}$ corresponding to $L$ such that $\forall n$, $C_n$ has depth at most $2$ and circuit size $|C_n| \leq n^p$.}, denoted by $\propertyP \in \mathcal{CO}(\ACp)$, if for every constant $\delta>0$, for every quantum algorithm $\mathcal{A}$ that computes $\propertyP$ in the black-box setting, $\forall n' \in \mathbb{N}, \exists n \geq n'$ and $\exists$ a set $L=\{L^1, L^2, \ldots\} \subseteq \ACp$ of `hard languages', such that $\forall$ circuit families $\{C^1_{n''}\}_{n'' \in \mathbb{N}}$ corresponding to $L^1$, $\forall$ circuit families $\{C^2_{n''}\}_{n'' \in \mathbb{N}}$ corresponding to $L^2$, $\ldots$, $\mathcal{A}$ uses at least $\Q(\propertyP_{n})^{1-\delta}$ quantum time on at least one of the inputs in $\{C^i_n\}_{i \in [|L|]}$. 

In particular, we say a property $\propertyP$ is \emph{$\AC$-compression-oblivious} (denoted by $\propertyP \in \mathcal{CO}(\AC)$) if $\exists p \in \mathbb{N}$ such that $\propertyP \in \mathcal{CO}(\ACp)$.
\end{defn}

With that, we can conjecture the following using the QSETH framework by~\cite{BPS21}.

\begin{conjecture}[\acQSETH{}, consequences of~\cite{BPS21}]
\label{conj:ACqseth}
Let $\propertyP$ be an \emph{$\AC$-compression-oblivious} property. The $\AC$-$\QSETH$ conjecture states that $\exists p \in \mathbb{N}$, such that for every constant $\delta>0$, for every quantum algorithm $\mathcal{A}$ that computes $\propertyP$ in the white-box setting, $\forall n' \in \mathbb{N},\exists n \geq n'$, $\exists$ a set $L=\{L^1, L^2, \ldots\} \subseteq \ACp$, $\forall $ circuit families $\{C^1_{n''}\}_{n''\in \mathbb{N}}$ corresponding to $L^1$, $\forall$ circuit families $\{C^2_{n''}\}_{n''\in \mathbb{N}}$ corresponding to $L^2$, $\ldots$, the algorithm $\mathcal{A}$ uses at least $\Q(\propertyP_{n})^{1-\delta}$ quantum time on at least one of the inputs in $\{C^i_n\}_{i \in [|L|]}$. \footnote{It is good to note that the original QSETH framework by~\cite{BPS21} allows us to consider formulas of more complicated classes. However, taking this complexity class as the class of all poly-sized \CNF{} or \DNF{} formulas suffices for the results presented in this paper. Therefore, using $\textsc{AC}^{0}_{2}$ to denote the class of all poly-sized $\CNF{}$ and $\DNF{}$ formulas, we define \acQSETH{} as the conjecture stated here and use this as the main conjecture of our paper.} 
\end{conjecture}


In informal terms, the notion of $\AC$-compression-obliviousness captures properties whose query complexity is a lower bound for the time complexity to compute the property even for truth tables of small \CNF/\DNF{} formulas. And, the $\AC$-$\QSETH$ conjecture states that having access to this succinct representation of the truth table, i.e., the description of the formula itself, must not help towards improving the computation time in computing these properties. The terms \emph{black-box} and \emph{white-box} setting refer to how the algorithm can access the truth table - in the former setting (i.e., the black-box setting as stated in \Cref{def:CompressionOblivious}) even though the input is a \CNF/\DNF{} formula, the algorithm is only allowed to evaluate the formula on required inputs and has no access to the description of the inputted \CNF/\DNF{} formula, and, in the latter setting (i.e., the white-box setting as stated in \Cref{conj:ACqseth}) the algorithm is allowed to use the description as well.

In comparison to the original QSETH paper (by Buhrman et al.~\cite{BPS21}) where the framework was introduced and applied to a more complex class of formulas,\footnote{The authors in \cite{BPS21} extensively used QSETH framework for branching programs or equivalently NC circuits to show non-trivial lower bounds for edit distance and longest common subsequence problems.} this paper instead serves as a guide to using QSETH for the lowest level of formulas, i.e., poly-sized \CNF{} and \DNF{} formulas, in a more elaborate fashion.

\input{TableOfCNFResults}
\input{TableOfResults}


\paragraph{Summary and technical overview}
In this paper, we use the QSETH framework (or precisely, \acQSETH{}) to `generate' natural variations of QSETH such as $\oplus$QSETH, $\#$QSETH, maj-QSETH, etc., which could (arguably) already be acceptable standalone conjectures in the quantum setting, and study some of their interesting implications. Additionally, we also use the QSETH framework to prove quantum lower bounds for \emph{approximately} counting the number of satisfying assignments to \CNF{} formulas, a problem whose complexity has been of interest in the classical setting \cite{DL21}; we study its quantum implications. See \Cref{sec:LowerBoundsVariants} for details. Proof of this result follows from a more detailed exploration of the QSETH framework than what was required in the original paper. Thus, as another contribution of this paper, we present a useful guide on how to effectively use the QSETH framework. Here we carefully summarize our contributions and technical overview below: 
\begin{itemize}
    \item We zoom into Buhrman et al.'s QSETH framework at the lowest-level formula class, i.e., the class of polynomial-size \CNF{}s and \DNF{}s, and use it to study the quantum complexity of variations of \CNFSAT{} problems. The QSETH framework is quite general which also makes it not entirely trivial to use it thus, we present a useful guide on how to effectively use the \acQSETH{} conjecture, for e.g., what lemmas need to be proved and what assumptions are needed to be made in order to understand the quantum complexity of \CNFSAT{} and its variants; see \Cref{fig:Flowchart}.
    
    \item We can categorise the several variants of \CNFSAT{} in two ways. First classification can be done by the width of the \CNF{} formulas, i.e., $k$-\CNF{}s versus \CNF{}s of unbounded clause length. Second classification is made with the choice of the property of the truth table one desires to compute. See the summary of complexity all \CNFSAT{} variants and their respective quantum time lower bounds in \Cref{table:SummaryOfCNFResults} and see below for the overview of the techniques used.
    \begin{itemize}
        \item To prove the quantum time lower bounds for the property variants of \CNFSAT{} problem we invoke \acQSETH{} (\Cref{conj:ACqseth}). But, \acQSETH{} conjectures the hardness of properties on a set of \CNF{} and \DNF{} formulas. For properties like \propertyCount{}, \propertyParity{}, \propertyMajority{}, etc., it easily follows from De Morgan's laws that these properties are equally hard on both \CNF{} and \DNF{} formulas. However, such arguments no longer hold when the properties are approximate variants of count for which we give nontrivial proofs; see \cref{sec:AdditiveErrorCNFSAT,sec:MultiplicativeFactorCNFSAT}. 

        \item Additionally, we also use \acQSETH{} to understand quantum complexity of $k$-SAT and its property variants. As a first step we study the classical reduction from \CNFSAT{} to $k$-SAT given by \cite{Calabro-DualityWidthClause-2006} and observe that the $2^{\frac{n}{2}}$ quantum lower bound for $\kSAT{}$, for $k=\Theta(\log n)$, follows from the quantum lower bound of \CNFSAT{}. Moreover, we make an important observation that this reduction by \cite{Calabro-DualityWidthClause-2006} is count-preserving\footnote{Should not be mistaken to be parsimonious, see \Cref{sec:CountingKSAT} for details.} and can be used to conclude lower bounds for other counting variants of $k$-SAT. See \Cref{table:SummaryOfCNFResults} for the summary of these bounds.
    \end{itemize}

    \item Having (somewhat) understood the complexities of the above-mentioned variants of \CNFSAT{}, we then prove conditional quantum time lower bounds for lattice problem, strong simulation, orthogonal vectors, set cover, hitting set problem, and their respective variants; see \Cref{table:SummaryOfResults}. 
    \begin{itemize}
        \item The quantum $2^{\frac{n}{2}}$ time lower bound we present for $\text{CVP}_p$ (for $p \notin 2\mathbb{Z}$) follows from a reduction from \kSAT{} to $\text{CVP}_p$ by \cite{BGS17,DBLP:conf/soda/AggarwalBGS21} and from the hardness result of $k$-SAT we present. Though such a result would also trivially follow by using Aaronson et al.'s version of QSETH, we stress that our hardness result of $k$-SAT is based on basic-QSETH which is a more believable conjecture.\footnote{If basic-QSETH from Buhrman et al.'s framework is false then Aaronson et al.'s QSETH is also false, but the implication in the other direction is not obvious.}
        
        \item Additionally, we also discuss the quantum complexity of the lattice counting problem (for non-even norm). We present a reduction, using a similar idea of \cite{BGS17}, from $\#k$-SAT to the lattice counting problem and we show a $2^n$ time quantum lower bound for the latter when $k=\Theta(\log n)$. As mentioned earlier, we get a $2^n$ time quantum lower bound for $\#k$-SAT, when $k=\Theta(\log n)$, using \acQSETH{}.
        
        \item As another application to the bounds we get from the property variants of \CNFSAT{} we look at the strong simulation problem. It was already established by \cite{CHM21, VDN10} that strong simulation of quantum circuit is a \#P-hard problem but in this work we give exact lower bounds for the same. Additionally, using the lower bounds of approximate counts of \CNFSAT{} we are able to shed light on how hard it is to quantumly solve the strong simulation problem with additive and multiplicative error approximation.

        \item Last but not least, we are also able to use the lower bounds for the property variants of \CNFSAT{} to give interesting lower bounds for orthogonal vectors, hitting set problem and their respective variants. See \cref{sec:OVandOtherResults} for more details. 
    \end{itemize}
    Our motivation to study the worst-case complexities of counting versions of these problems stems from the fact that worst-case complexity of counting versions of problems have been used in past to understand average-case complexity of other related problems. And, computational problems that have high average-case complexities usually find their place in cryptographic primitives. For example, Goldreich and Rothblum in \cite{GR18} present a worst-case to average-case reduction for counting $t$-cliques in graph and use the average-case hardness result towards constructing an interactive proof system. Another such example is that of the \OV{} problem - Ball, Rosen, Sabin and Vasudevan in \cite{BRSV-PoW18} use the worst-case hardness of the counting variant of \OV{} to first prove average-case hardness of evaluating low-degree polynomials which they use towards a Proofs of Work (PoW) protocol. Furthermore, Buhrman, Patro and Speelman in \cite{BPS21} observed that this PoW protocol in combination with QSETH ensures that the quantum provers also require the same time as the classical provers.\footnote{Note that counting the number of \OV{} pairs on \emph{average} has a fast algorithm \cite{DLW20averageFineGrained} so a worst-case to average-case reduction for counting \OV{} is not possible under standard fine-grained complexity assumptions.}
\end{itemize}

\paragraph{Related work}
Our paper is a follow-up work to the original QSETH paper by \cite{BPS21}; also the list of problems for which we show lower bounds does not overlap with the problems studied in~\cite{BPS21}. A basic version of QSETH was also introduced by Aaronson et al.\ \cite{BasicQSETH-Aaronson-2020} where they primarily used it to study the quantum complexity of closest pair and bichromatic pair problems; they also discuss the complexity of the (vanilla version of) orthogonal vector problem. Prior to this work, a quantum hitting-set conjecture was proposed and its implications were discussed in Schoneveld's bachelor thesis \cite{Sch22}, but their definition of hitting set is different from ours. In our work, we observe that the parsimonious reduction from \CNFSAT{} to hitting set (\Cref{def:HS}) given by \cite{CDLMNOPSW16} is easily quantizable, using which we get a QSETH-based lower bound. Recently, Huang et al.~\cite{huang2024finegrained} showed a significant barrier to establishing fine-grained quantum reductions from $k$-SAT to lattice problems in the Euclidean norm. In contrast, our work focuses on lattice problems in the $\ell_p$ norm, where $p$ is not an even integer.


\paragraph{Structure of our paper}
In \Cref{sec:LowerBoundsVariants}, we discuss the quantum complexity of variants of \CNFSAT{} problems, for e.g.\ \CountingCNFSAT{}, \ParityCNFSAT{}, \appCountingCNFSAT{}, etc., conditional on \acQSETH{}. The broad idea to show the hardness results of these variants is quite similar; the \emph{several} lemmas only account for the properties being different. Using results from \Cref{sec:LowerBoundsVariants}, we show quantum time lower bounds for strong simulation problem in \Cref{sec:HardnessStrongSimulation}, lattice problems in \Cref{sec:lattice}, and orthogonal vectors, hitting set and set cover in \Cref{sec:OVandOtherResults}. 

\begin{figure}
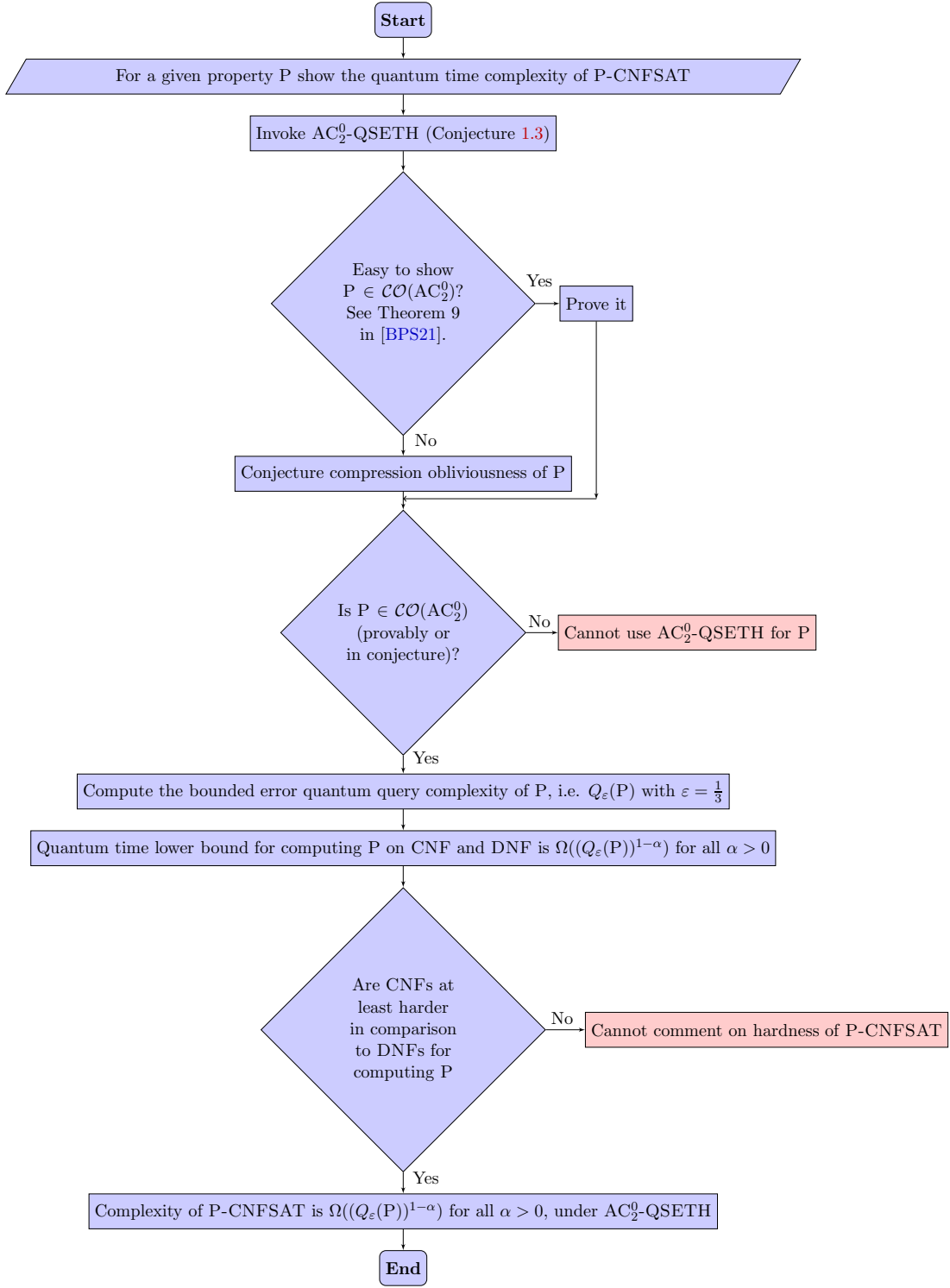

    \centering
    \include{Flowchart}
    \caption{Step-by-step guide on how to use the QSETH framework in a plug-and-play manner to show hardness results for $\textsc{P-}\CNFSAT{}$. Here $\textsc{P}$ can be any (partial or total) Boolean property of truth tables.}
    \label{fig:Flowchart}
\end{figure}

\newpage

\section{Preliminaries}

Throughout the paper, we use $n,d$ as natural numbers. 
For $1 \leq p < \infty$ the $\ell_p$-norm $\| \vec{x} \|_p$ of any vector $\vec{x} \in \R^d$ is defined by
\[
\| \vec{x} \|_p := \Big(\sum_{i = 1}^d \abs{x_i}^p\Big)^{1/p} \ .
\]
Additionally, the $\ell_{\infty}$ norm of such a vector $\vec{x}$ is defined as $\norm{\vec{x}}_{\infty} := \max_{i \in [d]} \abs{x_i}$.
The $\ell_p$ unit ball $B^d_p \subset \R^d$ is defined by $\{\vec{x} \in \R^d \mid \| \vec{x} \|_p \leq 1\}$.

\subsection{Quantum query complexity and quantum subroutines}

We start with defining Boolean properties, and then we will define the bounded-error quantum complexity of computing those properties.

\begin{defn}[Property]
\label{def:BooleanProperty}
We define a Boolean property (or just property) as a sequence $\propertyP \coloneqq\big( \propertyP_n\big)_{n\in\mathbb{N}}$ where each $\propertyP_n$ is a set of Boolean functions defined on $2^n$ variables.
\end{defn}

\paragraph{Quantum query complexity of $\propertyP$} The (bounded-error) quantum query complexity is defined only in a non-uniform setting, therefore, it is defined for $P_n$ for every $n \in \mathbb{N}$ instead of directly defining for $\propertyP$.  A quantum query
algorithm $\mathcal{A}$ for $\propertyP_n: \{0,1\}^{2^n} \rightarrow \{0,1\}$, on an input $x \in \{0,1\}^{2^n}$ begins in a fixed initial state $\ket{\psi_0}$, applies a sequence of unitaries $U_0, O_x, U_1, O_x, \ldots, U_T$, and performs a measurement whose outcome is denoted by $z$. Here, the initial state $\ket{\psi_0}$ and the unitaries $U_0, U_1, \ldots, U_T$ are independent of the input $x$. The unitary $O_x$ represents the ``query'' operation, and maps $\ket{i}\ket{b}$ to $\ket{i}\ket{b + x_i \textnormal{ mod } 2}$ for all $i \in [2^n]-1$. 
We say that $\mathcal{A}$ is a $1/3$-bounded-error algorithm computing $\propertyP_n$ if for all $x$ in the domain of $\propertyP_n$, the success probability of outputting $z=\propertyP_n(x)$ is at least $2/3$. Let cost$(\mathcal{A})$ denote the number of queries $\mathcal{A}$ makes to $O_x$ throughout the algorithm. The $1/3$-bounded-error quantum query complexity of $\propertyP_n$, denoted by $\Q(\propertyP_n)$, is defined as 
$\Q(\propertyP_n)=\min\{\text{cost}(\mathcal{A}):\mathcal{A} \text{ computes } \propertyP_n \text{ with error probability } \leq 1/3\}$.

We also introduce the quantum amplitude estimation subroutine.

\begin{theorem}[Implicit in Theorem~12 by \cite{BHMT02}]
\label{thm:AmplitudeEstimation}
Given a natural number $M$ and access to an $(n+1)$-qubit unitary $U$ satisfying
\begin{equation*}
    U\ket{0^n}\ket{0} = \sqrt{a}\ket{\phi_1}\ket{1}+\sqrt{1-a}\ket{\phi_0}\ket{0},
\end{equation*}
where $\ket{\phi_0}$ and $\ket{\phi_1}$ are arbitrary $n$-qubit states and $0<a<1$, there exists a quantum algorithm that uses $\mathcal{O}(M)$ applications of $U$ and $U^{\dagger}$ and $\widetilde {\mathcal{O}}(M)$ elementary gates, and outputs a state $\ket{\Lambda}$ such that after measuring that state, with probability $\geq 9/10$, the first register $\lambda$ of the outcome satisfies
\begin{equation*}
    |\sqrt{a}-\lambda|\leq \frac{100\pi}{M}.
\end{equation*}
\end{theorem}

\subsection{\CNFSAT{} and $k$-SAT}
A Boolean formula over variables $x_1,\ldots,x_n$ is in {\CNF{} form} if it is an AND of OR's of variables or their negations. More generally, a \CNF{} formula has the form 
$$\bigwedge\limits_{i}\Big(\bigvee\limits_{j} v_{ij}  \Big)$$
where $v_{ij}$ is either $x_k$ or $\neg x_k$. The terms $v_{ij}$ are called \emph{literals} of the formula and the disjunctions $\bigvee\limits_{j} v_{ij}$ are called its clauses. A $k$-CNF is a CNF formula in which all clauses contain at most $k$ literals (or the clause width is at most $k$). Note that when $k> n$, then clauses must contain duplicate or trivial literals (for example, $x_k\vee\neg x_k$ and $x_k\vee x_k$), therefore we can assume without loss of generality that $k$ is at most $n$. A DNF is defined in the exact same way as CNF, except that it is an OR of AND's of variables or their negations, that is, a DNF formula has the form $\bigvee\limits_{i}\Big(\bigwedge\limits_{j} v_{ij} \Big)$. 
We also define computational problems $k$-SAT and CNFSAT.

\begin{defn}[CNFSAT]
    Given as input a CNF formula $\phi$ defined on $n$ variables, 
    the goal is to determine if $\exists x\in\{0,1\}^n$ such that $\phi(x)=1$.
\end{defn}

\begin{defn}[$k$-SAT]
    Given as input a $k$-CNF formula $\phi$ defined on $n$ variables, 
    the goal is to determine if $\exists x\in\{0,1\}^n$ such that $\phi(x)=1$.
\end{defn}


\subsection{Lattice}
For any set of $n$ linearly independent vectors $\mathbf{B}=\{\vect{b_1},\hdots,\vect{b_n}\}$ from $\real^d$, the lattice $\cL$ generated by basis $\mathbf{B}$ is
\[ \cL(\mathbf{B})= \left\{ \sum\limits_{i=1}^n z_i\vect{b_i} : z_i \in \intg  \right\}.\]
We call $n$ the \textit{rank} of the lattice $\cL$ and $d$ the \textit{dimension}. The vectors $\mathbf{B}=\{\vect{b}_1,\ldots,\vect{b}_n\}$ forms a \textit{basis} of the lattice. Given a basis $\basis$, we use $\cL(\basis)$ to denote the lattice generated by $\basis$.

\begin{defn}
For any $1\leq p \leq \infty$, the Closest Vector Problem CVP$_p$ is the search problem defined as: The input is a basis $\basis \in \real^{d\times n}$ of the lattice $\cL$ and a target vector $\vect{t}$. The goal is to output a vector $\vect{v}\in \cL$ such that $\|\vect{v}-\vect{t}\|_p =  \min\limits_{x\in\cL}\|\vect{x}-\vect{t}\|_p$.
\end{defn}

With access to a CVP$_p$ solver, we can (almost) uniformly sparsify a lattice as follows.

\begin{theorem}[\cite{SteDiscreteGaussian16}, modified Theorem 3.2]\label{thm:spar}
Let ${\mathbf{B}}$ be a basis of a lattice $\cL({\mathbf{B}})$ and $Q$ be a prime number. Consider the sparsification process: input any two vectors $\vect{z},\vect{c}\in \intg_Q^n$, the sparsification process $Spar(\mathbf{B},Q,\vect{z},\vect{c})$ outputs a basis $\mathbf{B}_{Q,\vect{z}}$ of the sublattice $\cL_{\vect{z}}\subset \cL=\{\vect{u}\in \cL: \inProd{\vect{z}}{{\mathbf{B}}^{-1} \vect{u}}  =0 \mod Q\}$ and $\vect{w}_{\vect{z},\vect{c}}={\mathbf{B}}\vect{c}$. Then for every $\vect{t}\in \mathbb{R}^n$, $\vect{x}\in\cL$ with $N=|(\cL-\vect{t})\cap \|\vect{x}-\vect{t}\|\cdot B^n_p| \leq Q$, and $CVP_p$ oracle, we have
$$
\frac{1}{Q}-\frac{N}{Q^2}-\frac{N}{Q^{n-1}}\leq \Pr\limits_{\vect{z},\vect{c} \in \intg_Q^n}[CVP_p(\vect{t}+\vect{w}_{\vect{z},\vect{c}},\cL_{\vect{z}})=\vect{x}+\vect{w}_{\vect{z},\vect{c}}]\leq \frac{1}{Q}+\frac{1}{Q^n},
$$
and in particular,
$$
\frac{N}{Q}-\frac{N^2}{Q^2}-\frac{N^2}{Q^{n-1}}\leq \Pr\limits_{\vect{z},\vect{c} \in \intg_Q^n}[\min\limits_{\vect{u}\in\cL_{z}}\|(\vect{t}+\vect{w}_{\vect{z},\vect{c}}-\vect{u}\|_p)\leq\|\vect{x}-\vect{t}\|_p]\leq \frac{N}{Q}+\frac{N}{Q^n}.
$$
\end{theorem}

\section{Lower bounds for variants of CNFSAT using \acQSETH{}} 
\label{sec:LowerBoundsVariants}
We will now define several variants of \CNFSAT{} problem and using \acQSETH{} understand the quantum complexities of all of them. The consequences of these hardness results, some of which follow immediately and the rest with some work, will be discussed in \Cref{sec:HardnessStrongSimulation,sec:lattice,sec:OVandOtherResults}. We begin with some common variants of \CNFSAT{} problem (such as \kSAT{}) which are also very well studied classically \cite{CDLMNOPSW16}; we do this in \cref{sec:PopularVariantsCNFSAT}. And, proceed with some less popular variants (\cref{sec:AdditiveErrorCNFSAT,sec:MultiplicativeFactorCNFSAT,sec:CountingKSAT}) but with interesting consequences (presented in \cref{sec:HardnessStrongSimulation,sec:lattice,sec:OVandOtherResults}).





\subsection{Quantum complexity of \CNFSAT{} and other related problems}
\label{sec:basicCNFSAT}
We first restate the quantum hardness of \CNFSAT{} before delving into showing hardness results for its other variants. Interestingly, for the property $\propertyOR : \{0,1\}^{2^n} \rightarrow \{0,1\}$, where for $x \in \{0,1\}^{2^n}$ we define $\propertyOR(x)=1$ if $|x|\geq 1$ and $\propertyOR(x)=0$ whenever $|x|=0$, we can explicitly prove that $\propertyOR{}\in \mathcal{CO}(\AC)$~\cite{BPS21,CCKPS25Qnaturalproof}. 
Also, note that computing $\propertyOR{}$ on truth tables of \DNF{} formulas of $\poly(n)$ length can be computed in $\poly(n)$ time. Hence, using \acQSETH{} we can recover the following Basic-QSETH conjecture.

\begin{corollary}[\BasicQSETH{} \cite{BPS21}]\label{cor:basicQSETH}
 For each constant $\delta > 0$, there exists $c>0$ such that there is no bounded-error quantum algorithm that solves \CNFSAT{} (even restricted to formulas with $m\leq c n^2$ clauses) in $\mathcal{O}\left(2^{\frac{n(1-\delta)}{2}}\right)$ time, unless \acQSETH{} (\cref{conj:ACqseth}) is false.
\end{corollary}

Note that Aaronson et al.~\cite{BasicQSETH-Aaronson-2020} directly conjecture the above statement, while in our case the above conjecture is implied by  \acQSETH{} (\cref{conj:ACqseth}), and we will show how \acQSETH{} can imply other conjectured time lower bound for variants of CNFSAT problems in this subsection.

\subsubsection{Quantum complexity of \CountingCNFSAT{}, \ParityCNFSAT{}, \qParityCNFSAT{} and \majorityCNFSAT{}}
\label{sec:PopularVariantsCNFSAT}
To give conditional quantum lower bounds for variants of \CNFSAT{}, we should understand their corresponding quantum query lower bound (on the $2^n$-bit truth table). Here we introduce the properties that correspond to those popular variants of \CNFSAT{} (which will be defined later.)


\begin{defn} 
\label{def:PopularBooleanProperties} 
Let $|x|=|\{i : x_i=1\}|$ denote the Hamming weight of $N$-bit binary string $x$. We here list some properties defined on binary strings.
\begin{enumerate}
    \item \propertyCount{}: Let $\propertyCount :\{0,1\}^N \rightarrow [N]\cup\{0\}$ be the non-Boolean function defined by $\propertyCount{}(x) = |x|.$  
    \item \propertyParity{}: Let $\propertyParity :\{0,1\}^N \rightarrow \{0,1\}$ be the Boolean function defined by $\propertyParity{}(x) = |x| \bmod 2.$ 
    \item $\propertyCount_q$: Let $q$ be an integer and let $\propertyCount_q :\{0,1\}^N \rightarrow [q]-1$ be the non-Boolean function defined by $\propertyCount_q(x) = |x| \mod q$. 
    \item $\propertyMajority$: Let $\propertyMajority :\{0,1\}^N \rightarrow \{0,1\}$ be the Boolean function defined by 
\begin{equation*}
    \propertyMajority{}(x) =
\begin{cases}
  1  & \text{ if } |x|\geq \frac{N}{2}, \\
  0 & \text{otherwise.}
\end{cases}
\end{equation*}

\item[] And, there is also the following function almost similar to \propertyMajority{}.
\item $\propertyStrictMajority$: Let $\propertyStrictMajority :\{0,1\}^N \rightarrow \{0,1\}$ be the Boolean function with 
\begin{equation*}
    \propertyStrictMajority{}(x) =
\begin{cases}
  1  & \text{ if } |x| > \frac{N}{2}, \\
  0 & \text{otherwise.}
\end{cases}
\end{equation*}
\end{enumerate}  
\end{defn}

Here, we define variants of \CNFSAT{} corresponding to the above-mentioned properties. 
\begin{defn}[variants of CNFSAT] 
\label{def:PopularVariantsCNF}
Let $|\phi|=\{y\in\{0,1\}^n:\phi(y)=1\}$ denote the Hamming weight of the truth table of $\phi$. The following lists five variants of \CNFSAT{}: 
\begin{enumerate}
    \item \CountingCNFSAT{}: Given a \CNF{} formula $\phi$ on $n$ input variables, output $|\phi|$. 
    \item \ParityCNFSAT{}: Given a \CNF{} formula $\phi$ on $n$ input variables, output $|\phi|\bmod 2$. 
    \item \qParityCNFSAT{}: Given a \CNF{} formula $\phi$ on $n$ input variables and an integer $q \in [2^n] \setminus \{1\}$, output $|\phi|\bmod q$. 
    \item \majorityCNFSAT{}: Given a \CNF{} formula $\phi$ on $n$ input variables, output $1$ if $|\phi|\geq 2^{n}/2$ (else output $0$). 
    \item \stMajorityCNFSAT{}: Given a \CNF{} formula $\phi$ on $n$ input variables, output $1$ if $|\phi|> 2^{n}/2$ (else output $0$). 
\end{enumerate}
    
\end{defn}

Now again, we use the quantum query lower bound for $\propertyP$ whenever we want to discuss the time complexity of $\propertyP$-CNFSAT as in the QSETH framework by~\cite{BPS21}. Therefore, immediately after the definitions for variants of CNFSAT (with respect to property $\propertyP$), we will introduce the corresponding bounded-error quantum \emph{query} lower bound for computing $\propertyP$, and then conjecture the \emph{time} lower bound for $\propertyP$-CNFSAT ($\propertyP$ variant CNFSAT) using this query lower bound.
After that, we can use lower bounds for those variants of \CNFSAT{} to understand the quantum complexity of (variants of) $k$-SAT. We include the quantum query lower bounds for those properties for completeness. 

\begin{lemma}[\cite{BBCMW01}] The bounded-error quantum query complexity for \propertyCount{}, \propertyParity{}, \propertyMajority{} and \propertyStrictMajority{} on inputs of $N$-bit Boolean strings is $\Omega(N)$.    
\end{lemma}

\begin{proof}
\cite{BBCMW01} showed that the bounded-error quantum query complexity of a (total) Boolean function $f: \{0,1\}^N \rightarrow \{0,1\}$, denoted by $\Q(f)$ 
is lower bounded by $1/2$ of the degree of a minimum-degree polynomial $p$ that approximates $f$ on all $X \in \{0,1\}^N$, i.e., $|p(X)-f(X)|\leq 1/3$; let us denote this approximate degree by $\widetilde{deg}(f)$. Another important result by Paturi \cite{Paturi92} showed that if $f$ is a non-constant, symmetric\footnote{A symmetric Boolean function $f:\{0,1\}^N \rightarrow \{0,1\}$ implies $f(X)=f(Y)$ for all $X,Y$ whenever $|X|=|Y|$.} and total Boolean function on $\{0,1\}^N$ then $\widetilde{deg}(f)=\Theta(\sqrt{N(N-\Gamma(f))})$ where $\Gamma(f)=\min \{ |2k-N+1|: f_k\neq f_{k+1} \text{ and } 0\leq k \leq N-1\}$ and $f_k=f(X)$ for $|X|=k$. 

Using the above two results we can show the following:
\begin{enumerate}
    \item $\Gamma(\propertyParity)= 0$ for odd $N$ and $\Gamma(\propertyParity)=1$ whenever $N$ is even. Hence $\Q(\propertyParity)=\Omega(N)$.\footnote{One can actually immediately give $\Q(\propertyParity)\geq N/2$ by an elementary degree lower bound without using Paturi's result.}
    \item Similar to the above item $\Gamma(\propertyMajority)=\Gamma(\propertyStrictMajority)=0$ for odd $N$ and $\Gamma(\propertyMajority)=\Gamma(\propertyStrictMajority)=1$ otherwise. Hence, $\Q(\propertyMajority)=\Omega(N)$ and $\Q(\propertyStrictMajority)=\Omega(N)$.
    \item Any of the above three properties can be computed from \propertyCount{}. Hence, 
    $\Q(\propertyCount)=\Omega(N)$.
\end{enumerate} 
\end{proof}

\begin{lemma}\label{thm:ModqQueryComplexity}
Let $q\in [3, \frac{N}{2}]$ be an integer and $\propertyCount_q : \{0,1\}^N \rightarrow [q]-1$ be the function defined by $\propertyCount_q(x)=\propertyCount(x) \mod q$. Then $\Q(\propertyCount_q)= \Omega(\sqrt{N(N-2q+1)})$.
\end{lemma}
\begin{proof}
Let $\textsc{dec-count}_q$ be a decision version of the $\propertyCount_q$ defined for all $x \in \{0,1\}^N$ as
\begin{equation}
    \textsc{dec-count}_q(x)=\begin{cases}
    1, \text{ if } \propertyCount_q(x)=q-1,\\
    0, \text{ otherwise}.
    \end{cases}
\end{equation} 

When the function is non-constant and symmetric then one can use Paturi's theorem to characterize the approximate degree of that function \cite{Paturi92}. It is easy to see that $\textsc{dec-count}_q$ is a non-constant symmetric function. Combining both these results we get that $\Q(\textsc{dec-count}_q)=\Omega(\sqrt{N(N-\Gamma(\textsc{dec-count}_q))})$.

We now compute the value of $\Gamma(\textsc{dec-count}_q)$. 
For any symmetric Boolean function $f:\{0,1\}^N \rightarrow \{0,1\}$ the quantity $\Gamma(f)$ is defined as $\Gamma(f)=\min_{k} \{|2k-N+1|\}$ such that $f_k\neq f_{k+1}$ and $f_k=f(x)$ for $|x|=k$ with $1 \leq k \leq N-1$. It is easy to see that $\textsc{dec-count}_q(x)=1$ only for $x$ with Hamming weight $|x|=rq-1$ where $r$ is an integer and $\textsc{dec-count}_q(x)=0$ elsewhere. Let $r'$ be the integer such that $r'q-1 \leq \frac{N}{2} \leq (r'+1)q-1$ then the $k$ minimizing $\Gamma(\textsc{dec-count}_q)$ is either $r'q-1$ or $(r'+1)q-1$. This implies that $\Gamma(\textsc{dec-count}_q) \leq 2q-1$, which in turn implies that $N-\Gamma(\textsc{dec-count}_q) \geq N-2q+1$. Therefore, $\Q(\textsc{dec-count}_q)=\Omega(\sqrt{N(N-2q+1)})$.

As one can compute $\textsc{dec-count}_q$ using an algorithm that computes $\propertyCount_q{}$, we therefore have $Q(\propertyCount_q{})=\Omega(\sqrt{N(N-2q+1)})$.
\end{proof}


As we don't yet know how to prove compression-obliviousness of properties with high query complexities (Theorem~9 in \cite{BPS21}) we instead conjecture that \propertyCount{}, \propertyParity{}, \propertyMajority{} and \propertyStrictMajority{} are compression oblivious for poly-sized \CNF{} and \DNF{} formulas. We think it is reasonable to make this conjecture since it will falsify certain commonly-used cryptography assumptions if those properties are not compression oblivious. See~\cite{CCKPS25Qnaturalproof} for a discussion on this topic. 

\begin{conjecture}
\label{conj:CompressionObliviousnessAC}
The following properties
\begin{enumerate}
    \item $\propertyParity{}:\{0,1\}^{2^n} \rightarrow \{0,1\}$, 
    \item $\propertyCount_q :\{0,1\}^{2^n} \rightarrow [q-1]\cup \{0\}$ where $2 < q \leq 2^{n-1}$,
    \item $\propertyMajority{}:\{0,1\}^{2^n} \rightarrow \{0,1\}$, and 
    \item $\propertyStrictMajority{}:\{0,1\}^{2^n} \rightarrow \{0,1\}$
\end{enumerate}
stated in \Cref{def:PopularBooleanProperties} are in $\mathcal{CO}(\textsc{AC}^{0}_{2})$.
\end{conjecture}

\begin{corollary}
Let $\textsc{AC}^{0}_{2}$ denote the class of $\poly(n)$ sized \CNF{} and \DNF{} formulas on $n$ input variables. If any one item of \cref{conj:CompressionObliviousnessAC} is true then the property $\propertyCount{}:\{0,1\}^{2^n}\rightarrow [2^n] \cup \{0\}$ is in $\mathcal{CO}(\textsc{AC}^{0}_{2})$.
\end{corollary}

We can now invoke \acQSETH{} (as mentioned in \cref{conj:ACqseth}) to prove the quantum hardness for \CountingCNFSAT{}, \ParityCNFSAT{}, \majorityCNFSAT{} and \stMajorityCNFSAT{}.

\begin{theorem}[\#QSETH] 
\label{thm:countingQSETH}
For each constant $\delta>0$, there exists $c>0$ such that there is no bounded-error quantum algorithm that solves \CountingCNFSAT{} (even restricted to formulas with $m\leq c n^2$ clauses)  in $\mathcal{O}(2^{n(1-\delta)})$ time, unless \acQSETH{} (\Cref{conj:ACqseth}) is false or $\propertyCount \notin \mathcal{CO}(\textsc{AC}^{0}_{2})$ (i.e., each item of \Cref{conj:CompressionObliviousnessAC} is false).
\end{theorem} 

\begin{proof}
By way of contradiction, let us assume that there exists a bounded-error quantum algorithm~$\mathcal{A}$ that solves \CountingCNFSAT{} on $n$ variables (and on $m$ clauses with some $m\leq cn^2$) in $\mathcal{O}(2^{n(1-\delta)})$ time for some $\delta>0$. Then given a circuit $C \in \textsc{AC}^{0}_{2}$ we do one of the following:
\begin{itemize}
    \item if $C$ is a poly-sized \CNF{} formula then we use the algorithm $\mathcal{A}$ to compute the number of satisfying assignments to $C$ in $\mathcal{O}(2^{n(1-\delta)})$ time. Or,
    \item if $C$ is a poly-sized \DNF{} formula then we first construct the negation of $C$, let us denote by $\neg C$, in $\poly(n)$ time; using De Morgan's law we can see that the resulting formula $\neg C$ will be a $\poly(n)$ \CNF{} formula. Using $\mathcal{A}$ we can now compute the number of satisfying assignments $t$ to $\neg C$ in $\mathcal{O}(2^{n(1-\delta)})$ time. The number of satisfying assignments to $C$ would be then $2^n-t$.
\end{itemize}
The existence of an algorithm such as $\mathcal{A}$ would imply that \acQSETH{} is false. Hence, proved.
\end{proof}

Using similar arguments as in the proof of \cref{thm:countingQSETH} we can conclude the following statements.

\begin{corollary}[$\oplus$QSETH]
\label{thm:ParityQSETH}
 For each constant $\delta>0$, there exists $c>0$ such that there is no bounded-error quantum algorithm that solves \ParityCNFSAT{} (even restricted to formulas with $m\leq c n^2$ clauses)  in $\mathcal{O}(2^{n(1-\delta)})$ time,  unless \acQSETH{} (\Cref{conj:ACqseth}) is false or $\propertyParity{} \notin \mathcal{CO}(\textsc{AC}^{0}_{2})$ (i.e., Item 1 of \Cref{conj:CompressionObliviousnessAC} is false).   
\end{corollary}

\begin{corollary}[$\#_q$QSETH]
\label{conj:qParitySAT}
Let $q\in [3, \frac{N}{2}]$ be an integer. For each constant $\delta>0$, there exists $c>0$ such that there is no bounded-error quantum algorithm that solves \qParityCNFSAT{} (even restricted to formulas with $m\leq c n^2$ clauses) in $\mathcal{O}\left( {2}^{n(1-\delta)}\right)$ time, unless \acQSETH{} (\Cref{conj:ACqseth}) is false or $\propertyCount_q \notin \mathcal{CO}(\textsc{AC}^{0}_{2})$ (i.e., Item 2 of \Cref{conj:CompressionObliviousnessAC} is false).
\end{corollary}

\begin{corollary}[Majority-QSETH]
\label{thm:MajorityQSETH}
 For each constant $\delta>0$, there exists $c>0$ such that there is no bounded-error quantum algorithm that solves
\begin{enumerate}
    \item \majorityCNFSAT{} (even restricted to formulas with $m\leq c n^2$ clauses)  in $\mathcal{O}(2^{n(1-\delta)})$ time,  unless \acQSETH{} is false or $\propertyMajority \notin \mathcal{CO}(\textsc{AC}^{0}_{2})$ (i.e., Item 3 of \Cref{conj:CompressionObliviousnessAC} is false);
    \item \stMajorityCNFSAT{} (even restricted to formulas with $m\leq c n^2$ clauses)  in $\mathcal{O}(2^{n(1-\delta)})$ time,  unless \acQSETH{} is false or $\propertyStrictMajority \notin \mathcal{CO}(\textsc{AC}^{0}_{2})$ (i.e., Item 4 of \Cref{conj:CompressionObliviousnessAC} is false).
\end{enumerate}
\end{corollary}

Akmal and Williams showed that one can actually compute the Majority on the truth table of $k$-CNF formulas for constant $k$ in polynomial time, while computing the strict-Majority on the truth table of such formulas is NP-hard~\cite{AW21}. Therefore, here we define both majority and strict-majority and their variants of CNFSAT problems for clarity (and state the hardness of both problems in one conjecture). Note that for CNFSAT, each clause is allowed to contain $n$ literals (which means $k$ is no longer a constant), and in this case, it is not clear if one can solve \majorityCNFSAT{} in polynomial time or not. Therefore, none of \acQSETH{}, Items 3 and 4 of \Cref{conj:CompressionObliviousnessAC} is immediately false yet. (See also the discussion at the bottom of page~5 in the arXiv version of~\cite{AW21} for reductions between \majorityCNFSAT{} and \stMajorityCNFSAT{}.)


\subsubsection{Quantum complexity of \addErrorCountingCNFSAT{}}
\label{sec:AdditiveErrorCNFSAT}
Instead of the exact number of satisfying assignments to a formula, one might be interested in an additive-error approximation. Towards that, we define the problem \addErrorCountingCNFSAT{} as follows.

\begin{defn}[\addErrorCountingCNFSAT{}]
\label{def:AddErrorCountCNFSAT} Given a CNF formula $\phi$ 
on $n$ variables. 
The goal of the problem is to output 
an integer $d$ such that $|d-|\phi||< \Delta$ where 
$\Delta \in [1,2^n]$.
\end{defn}

This problem (\cref{def:AddErrorCountCNFSAT}) can be viewed as computing the following property on the truth table of the given formula.

\begin{defn}[\propertyAddErrorCount{}]
\label{def:AddErrorCount} 
Given a Boolean string $x \in \{0,1\}^N$, 
$\propertyAddErrorCount$ asks to output an integer $w$ such that $|w-|x|| < \Delta$ where 
$\Delta \in [1,N)$.
\end{defn}

Note that $\propertyAddErrorCount$ is a relation instead of a function now because its value is not necessarily uniquely defined. The bounded-error quantum query complexity for computing \propertyAddErrorCount{} was studied in \cite{NW99}. They showed the following result.

\begin{theorem}[Theorem~1.11 in \cite{NW99}] Let  $\Delta \in [1,N)$. Every bounded-error quantum algorithm that computes $\propertyAddErrorCount$ uses $\Omega \left(\sqrt{\frac{N}{\Delta}}+\frac{\sqrt{t(N-t)}}{\Delta} \right)$ quantum queries on inputs with $t$ ones.
\end{theorem}

For values of $\Delta=o(\sqrt{t})$ we are unable to prove the compression-obliviousness of this property. Hence, we make the following conjecture.

\begin{conjecture}
\label{conj:CompressionObOfAddErrorCount}
$\propertyAddErrorCount \in \mathcal{CO}(\textsc{AC}^{0}_{2})$.   
\end{conjecture}

One can now establish the time lower bound for computing the \propertyAddErrorCount{} on $\poly(n)$-sized \CNF{} and \DNF{} formulas. However, this doesn't automatically imply the same lower bound for the case when there are only \CNF{} formulas to consider. Fortunately, \propertyAddErrorCount{} is defined in such a way that computing this property is equally hard for both \CNF{} and \DNF{} formulas. More precisely, the following statement holds.

\begin{theorem}[$\Delta$\textsc{-add-}\#QSETH]
\label{thm:AddErrorQSETH}
Let $\Delta \in [{1},2^n)$. For each constant $\delta>0$, there exists $c>0$ such that there is no bounded-error quantum algorithm that solves \addErrorCountingCNFSAT{} (even restricted to formulas with $m\leq c n^2$ clauses) in $\mathcal{O}\left( \left(\sqrt{\frac{N}{\Delta}}+\frac{\sqrt{\hat{h}(N-\hat{h})}}{\Delta} \right)^{1-\delta}\right)$ time where $\hat{h}$ is the number of satisfying assignments, unless \acQSETH{} (\cref{conj:ACqseth}) is false or $\propertyAddErrorCount \notin \mathcal{CO}(\textsc{AC}^{0}_{2})$ (i.e., \cref{conj:CompressionObOfAddErrorCount} is false). 
\end{theorem}

\begin{proof}
By way of contradiction let's assume that there is an algorithm $\mathcal{A}$ such that given a \CNF{} formula it can compute the \propertyAddErrorCount{} on its truth table in  $\mathcal{O}\left( \left(\sqrt{\frac{N}{\Delta}}+\frac{\sqrt{t(N-t)}}{\Delta} \right)^{1-\beta}\right)$ time for some constant $\beta>0$.

Then, given a $\poly(n)$ sized \DNF{} formula on $n$ variables, let us denote that by $\phi$, we can run Algorithm $\mathcal{A}$ on $\neg \phi$ and use its output which is a $\Delta$ additive error approximation of the number of satisfying assignments to $\neg \phi$ to compute a $\Delta$ additive error approximation of the number of satisfying assignments to $\phi$.

Let us denote the number of satisfying assignments of $\neg \phi$ by $d'$ and the output of Algorithm $\mathcal{A}$ by $d$. This means we have $|d-d'| < \Delta$. We claim that $2^n -d$ will be a $\Delta$ additive error approximation of $2^n -d'$, which is the number of satisfying assignments of $\phi$; $|(2^n -d)-(2^n -d')|=|d'-d| < \Delta$.

Therefore, a $\mathcal{O}\left( \left(\sqrt{\frac{N}{\Delta}}+\frac{\sqrt{t(N-t)}}{\Delta} \right)^{1-\beta}\right)$ time algorithm for computing \propertyAddErrorCount{} on truth table of \CNF{} formulas also implies a $\mathcal{O}\left( \left(\sqrt{\frac{N}{\Delta}}+\frac{\sqrt{t(N-t)}}{\Delta} \right)^{1-\beta}\right)$ time algorithm for computing \propertyAddErrorCount{} on truth table of \DNF{} formulas; this violates the combination of \acQSETH{} and \cref{conj:CompressionObOfAddErrorCount}. Hence, proved.
\end{proof}





\subsubsection{Quantum complexity of \appCountingCNFSAT{} and other related problems}
\label{sec:MultiplicativeFactorCNFSAT}

One other approximation of the count of satisfying assignments is the multiplicative-factor approximation, defined as follows.
\begin{defn}[\appCountingCNFSAT{}]
\label{def:ApproximateCountCNFSAT}
Let $\gamma \in (0,1)$. The \appCountingCNFSAT{} problem is defined as follows. Given a CNF formula formula $\phi$ on $n$ Boolean variables, 
The goal of the problem is to output an integer $d$ such that $(1-\gamma)|\phi|< d < (1+\gamma) |\phi|$. 
\footnote{The same results hold if the approximation is defined with the equalities, i.e., $(1-\gamma)|\phi|\leq d \leq (1+\gamma) |\phi|$. An additional observation under this changed definition of \appCountingCNFSAT{} is as follows. Given a CNF formula as input, the algorithm for \appCountingCNFSAT{} outputs $0$ only when there is no satisfying assignment to that formula. Hence, one can decide satisfiability of a given CNF formula using the algorithm for \appCountingCNFSAT{}. Therefore, the same lower bound 
holds for this changed definition of \appCountingCNFSAT{}.}
\end{defn}

The expression $(1-\gamma)|\phi| < d < (1+\gamma)|\phi|$ 
can be categorized into the following two cases.
\begin{itemize}
    \item Case 1 is when $\gamma |\phi|\leq 1$: in this regime, the algorithm solving \appCountingCNFSAT{} is expected to return the value $|\phi|$, which is the \emph{exact} count of the number of solutions to the \CNFSAT{} problem. From \cref{thm:countingQSETH} we postulate that for each constant $\delta>0$, there is no $\mathcal{O}(2^{n(1-\delta)})$ time algorithm that can compute the exact number of solutions to input \CNF{} formula; this is a tight lower bound.
    
    \item Case 2 is when $\gamma |\phi| > 1$: in this regime, the algorithm solving \appCountingCNFSAT{} is expected to return a value $d$ which is a \emph{$\gamma$-approximate relative count} of the number of solutions to the \CNFSAT{} problem. 
\end{itemize}

In order to understand the hardness of \appCountingCNFSAT{} in the second case, we will first try to understand how hard it is to compute the following property. 
Let $f_{\ell,\ell'}:\mathcal{D} \rightarrow \{0,1\}$ with $\mathcal{D} \subset \{0,1\}^N$ be a partial function defined as follows
\begin{equation*}
f_{\ell,\ell'}(x)=
\begin{cases}
1, \text{ if } |x|=\ell,\\
0, \text{ if } |x|=\ell'.  
\end{cases}
\end{equation*}
Nayak and Wu in \cite{NW99} analyzed the approximate degree of $f_{\ell,\ell'}$. By using the polynomial method~\cite{BBCMW01} again 
we have a lower bound on the quantum query complexity of $f_{\ell,\ell'}$ as mentioned in the following statement.

\begin{lemma}\cite[Corollary~1.2]{NW99}
\label{cor:QueryComplexityEllFunction}
Let $\ell, \ell'\in\mathbb{N}$ be such that $\ell\neq \ell'$,  $f_{\ell,\ell'}:\mathcal{D} \rightarrow \{0,1\}$ where $\mathcal{D} \subset \{0,1\}^N$, and
\begin{equation*}
f_{\ell,\ell'}(x)=
\begin{cases}
1, \text{ if } |x|=\ell,\\
0, \text{ if } |x|=\ell'.  
\end{cases}
\end{equation*} Let $\Delta_{\ell}=|\ell-\ell'|$ and $p \in \{\ell, \ell'\}$ be such that $|\frac{N}{2}-p|$ is maximized. Then every bounded-error quantum algorithm that computes $f_{\ell,\ell'}$
uses $\Omega \left( \sqrt{\frac{N}{\Delta_{\ell}}}+\frac{\sqrt{p(N-p)}}{\Delta_{\ell}}\right)$ queries.    
\end{lemma}

Using \acQSETH{} we will now show that for a choice of $\ell,\ell'$ computing $f_{\ell,\ell'}$ on truth tables of $\CNF{}$ formulas of $\poly(n)$ size requires $\Omega \left( \sqrt{\frac{2^n}{\Delta_{\ell}}}+\frac{\sqrt{p(2^n-p)}}{\Delta_{\ell}}\right)$ time where $\Delta_{\ell}=|\ell-\ell'|$ and $p \in \{\ell,\ell'\}$ that maximises $|2^{n-1}-p|$. The only caveat (as also witnessed several times earlier) is that we cannot prove the compression obliviousness of $f_{\ell,\ell'}$ hence we state and use the following conjecture.

\begin{conjecture}
\label{conj:CompressionObOfEll}
For any pair of integers $\ell,\ell' \in [2^n] \cup \{0\}$ satisfying that $\ell\neq \ell'$, $f_{\ell,\ell'} \in \mathcal{CO}(\textsc{AC}^{0}_{2})$.\footnote{Note that, there are some values of $\ell, \ell'$ for which $f_{\ell,\ell'}$ will be provably compression oblivious, for e.g., $\ell=1$ and $\ell'=0$ would capture the OR property which is compression oblivious; see \cref{sec:basicCNFSAT}.}
\end{conjecture}

And we can show the following.

\begin{lemma}
\label{thm:AppQSETHwithEll}\label{thm:hardnessForSpecificEll}
Let $\ell, \ell' \in [2^n] \cup \{0\}$ be such that $\ell\neq \ell'$. If both \acQSETH{}(\cref{conj:ACqseth}) and  \Cref{conj:CompressionObOfEll} hold, then at least one of the following is true:
\begin{itemize}
    \item For each constant $\delta>0$, there exists $c>0$ such that there is no bounded-error quantum algorithm that computes $f_{\ell, \ell'}$ on the truth table of \CNF{} formulas defined on $n$ variables in $\mathcal{O}\left( \left( \sqrt{\frac{2^n}{\Delta_{\ell}}}+\frac{\sqrt{p(2^n-p)}}{\Delta_{\ell}}\right)^{1-\delta} \right)$ time  (even restricted to formulas with $m\leq c n^2$ clauses);
    \item For each constant $\delta>0$, there exists $c>0$ such that there is no bounded-error quantum algorithm that computes $f_{N-\ell, N-\ell'}$ on the truth table of \CNF{} formulas defined on $n$ variables in $\mathcal{O}\left( \left( \sqrt{\frac{2^n}{\Delta_{\ell}}}+\frac{\sqrt{p(2^n-p)}}{\Delta_{\ell}}\right)^{1-\delta} \right)$ time (even restricted to formulas with $m\leq c n^2$ clauses);
\end{itemize} 
here $\Delta_{\ell}=|\ell-\ell'|$ and $p \in \{\ell,\ell'\}$ such that $|2^{n-1}-p|$ is maximized.
In particular, when $\ell+\ell'=2^n$, the above immediately implies the following: 
\begin{itemize}
    \item For each constant $\delta>0$, there exists $c>0$ such that there is no bounded-error quantum algorithm that computes $f_{\ell,\ell'}$ on the truth table of \CNF{} formulas defined on $n$ variables in $\mathcal{O} \left( \left( \sqrt{\frac{2^n}{\Delta_{\ell}}}+\frac{\sqrt{\ell \ell'}}{\Delta_{\ell}}\right)^{1-\delta} \right)$ time (even restricted to formulas with $m\leq c n^2$ clauses).
\end{itemize}
\end{lemma}

\begin{proof}Let $N$ be an integer that we will fix later and let $f'_{\ell,\ell'}:\{0,1\}^{N} \rightarrow \{0,1\}$ be defined as follows
\begin{equation*}
f'_{\ell,\ell'}=
\begin{cases}
1, \text{ if } |x|=N-\ell,\\
0, \text{ if } |x|=N-\ell'.  
\end{cases}
\end{equation*}
It is not hard to see $f'_{\ell,\ell'}$ is the same as function $f_{N-\ell,N-\ell'}$. Fortunately, both the functions $f_{N-\ell,N-\ell'}$ and $f_{\ell,\ell'}$ have the same value of $\Delta_{\ell}$ and $h$ where $h=p(N-p)$. Therefore the bounded error quantum query complexity of $f'_{\ell,\ell'}$ is $\Omega \left( \sqrt{\frac{N}{\Delta_{\ell}}}+\frac{\sqrt{p(N-p)}}{\Delta_{\ell}}\right)$ where $\Delta_{\ell}=|\ell-\ell'|$ and $p \in \{\ell,\ell'\}$ such that $|\frac{N}{2}-p|$ is maximised; same as the bounded error quantum query complexity of $f_{\ell,\ell'}$ as mentioned in \cref{cor:QueryComplexityEllFunction}. 

Moreover, as $f'_{\ell,\ell'}$ is the same function $f_{N-\ell,N-\ell'}$ it is therefore clear from \cref{conj:CompressionObOfEll} that $f'_{\ell,\ell'} \in \mathcal{CO}(\textsc{AC}^{0}_{2})$ which means there is no bounded error quantum algorithm that can compute $f'_{\ell,\ell'}$ or $f_{\ell,\ell'}$ on truth tables of $\poly(n)$ size \CNF{} or \DNF{} formulas defined on $n$ input variables in $\mathcal{O} \left( \left( \sqrt{\frac{2^n}{\Delta_{\ell}}}+\frac{\sqrt{p(2^n-p)}}{\Delta_{\ell}}\right)^{1-\delta} \right)$ time for any constant $\delta>0$ unless \acQSETH{} is false. We will now show that conditional on \acQSETH{} this result holds even when we restrict ourselves to only $\poly(n)$ sized \CNF{} formulas.

Having introduced $f'_{\ell, \ell'}$ we will now prove \cref{thm:AppQSETHwithEll} using the following propositions.
\begin{itemize}
    \item \textbf{Proposition A} There is no bounded error quantum algorithm that can compute $f_{\ell,\ell'}$ on truth table of \CNF{} formulas defined on $n$ variables in $\mathcal{O} \left( \left( \sqrt{\frac{2^n}{\Delta_{\ell}}}+\frac{\sqrt{p(2^n-p)}}{\Delta_{\ell}}\right)^{1-\delta} \right)$ time for any $\delta>0$.
    \item \textbf{Proposition B} There is no bounded error quantum algorithm that can compute $f_{\ell,\ell'}$ on truth table of \DNF{} formulas defined on $n$ variables in $\mathcal{O} \left( \left( \sqrt{\frac{2^n}{\Delta_{\ell}}}+\frac{\sqrt{p(2^n-p)}}{\Delta_{\ell}}\right)^{1-\delta} \right)$ time for any $\delta>0$.
    \item \textbf{Proposition C} There is no bounded error quantum algorithm that can compute $f'_{\ell,\ell'}$ on truth table of \CNF{} formulas defined on $n$ variables in $\mathcal{O} \left( \left( \sqrt{\frac{2^n}{\Delta_{\ell}}}+\frac{\sqrt{p(2^n-p)}}{\Delta_{\ell}}\right)^{1-\delta} \right)$ time for any $\delta>0$.
    \item \textbf{Proposition D} There is no bounded error quantum algorithm that can compute $f'_{\ell,\ell'}$ on truth table of \DNF{} formulas defined on $n$ variables in $\mathcal{O} \left( \left( \sqrt{\frac{2^n}{\Delta_{\ell}}}+\frac{\sqrt{p(2^n-p)}}{\Delta_{\ell}}\right)^{1-\delta} \right)$ time for any $\delta>0$.
\end{itemize}
Conditional on \cref{conj:CompressionObOfEll} and \acQSETH{} the following statements hold.
\begin{itemize}
    \item \textbf{Claim 1} At least one of the propositions A or B is true.
    \item \textbf{Claim 2} At least one of the propositions C or D is true.
    \item \textbf{Claim 3} At least one of the propositions A or C is true; by way of contradiction let us assume that both propositions A and C are false, this means there exist algorithms $\mathcal{A}, \mathcal{A}'$ that for an $\delta>0$ and $\delta'>0$ compute $f_{\ell,\ell'}$ and $f'_{\ell,\ell'}$ on the truth table of $\poly(n)$ size \CNF{} formulas defined on $n$ input variables in $\mathcal{O} \left( \left( \sqrt{\frac{2^n}{\Delta_{\ell}}}+\frac{\sqrt{p(2^n-p)}}{\Delta_{\ell}}\right)^{1-\delta} \right)$ time and in $\mathcal{O} \left( \left( \sqrt{\frac{2^n}{\Delta_{\ell}}}+\frac{\sqrt{p(2^n-p)}}{\Delta_{\ell}}\right)^{1-\delta'} \right)$ time, respectively. Moreover, if propositions A and C are false then from Claims 1 and 2 we can deduce that both B and D must be true which means there is no quantum algorithm that can compute $f_{\ell,\ell'}$ or $f'_{\ell,\ell'}$ on the truth table of $\poly(n)$ size \DNF{} formulas on $n$ input variables in $\mathcal{O} \left( \left( \sqrt{\frac{2^n}{\Delta_{\ell}}}+\frac{\sqrt{p(2^n-p)}}{\Delta_{\ell}}\right)^{1-\delta} \right)$ time for any $\delta>0$. However, given a $\DNF{}$ formula $\phi$ as an input to compute $f_{\ell,\ell'}$ on its truth table one can instead compute $f'_{\ell,\ell'}$ on the negation of $\phi$, let us denote by $\neg \phi$, using algorithm $\mathcal{A}'$ on $\neg \phi$ in $\mathcal{O} \left( \left( \sqrt{\frac{2^n}{\Delta_{\ell}}}+\frac{\sqrt{p(2^n-p)}}{\Delta_{\ell}}\right)^{1-\delta'} \right)$ time which is a contradiction. This means at least one of the two propositions A or C must be true which is exactly the statement of \cref{thm:AppQSETHwithEll}.\qedhere
\end{itemize}
\end{proof}



Inspired by the arguments used in the proof of Theorem~1.13 in \cite{NW99}, we will now show that \cref{thm:hardnessForSpecificEll} implies the following result. Our result holds for $\gamma \in \left[\frac{1}{2^{n}},0.4999\right)$; this range of $\gamma$ suffices for our reductions presented in the later sections.

\begin{corollary}[$\gamma$-\#QSETH]
\label{thm:AppCountQSETH}
Let $\gamma \in \left[\frac{1}{2^{n}},0.4999\right)$.  For each constant $\delta>0$, there exists $c>0$ such that there is no bounded-error quantum algorithm that solves \appCountingCNFSAT{} (even restricted to formulas with $m\leq c n^2$ clauses) in time
\begin{enumerate}
    \item $\mathcal{O}\left( \left( \frac{1}{\gamma}\sqrt{\frac{2^{n}-\hat{h}}{\hat{h}}}\right)^{1-\delta}\right)$ if $\gamma \hat{h} > 1$, where $\hat{h}$ is the number of satisfying assignments;
    \item $\mathcal{O}(2^{n(1-\delta)})$ otherwise,
\end{enumerate}  unless \acQSETH{}(\cref{conj:ACqseth}) is false  or  $\ell\neq \ell'$, $f_{\ell,\ell'} \notin \mathcal{CO}(\textsc{AC}^{0}_{2})$ (i.e., \Cref{conj:CompressionObOfEll} is false).
\end{corollary}

We show the first part of \cref{thm:AppCountQSETH} in the following way and use the result from \cref{thm:countingQSETH} for the second part. Given a value of $\gamma \in \left[\frac{1}{2^{n}},0.4999\right)$  we will fix values of $\ell \in [2^n]\cup \{0\}$ and $\ell' \in [2^n]\cup \{0\}$ such that we are able to compute $f_{\ell,\ell'}$ on the truth table of an input \CNF{} formulas on $n$ variables 
using the algorithm that solves \appCountingCNFSAT{}. Hence, we can show a lower bound on \appCountingCNFSAT{} using the lower bound result from \cref{thm:hardnessForSpecificEll}. 

\begin{proof}
Let $N=2^n$. Let $\ell=\frac{N}{2}+\ceil{\gamma t}=\left\lceil{\frac{N}{2}+\gamma t}\right\rceil$ and $\ell'=\frac{N}{2}-\ceil{\gamma t}=\left\lfloor{\frac{N}{2}-\gamma t}\right\rfloor$; here $t \in [N]$ is a value that we will fix later but in any case, we have $1 \leq \ceil{\gamma t} < \frac{N}{2}$. With that, we are ensured that  $\gamma \ell > \frac{1}{2}$. We also make sure to choose values $\ell, \ell'$ in such a way that $\gamma \ell'=\Omega(1)$. Clearly, $\ell+\ell'=N$ and $\Delta_{\ell}=|\ell-\ell'|=2\ceil{\gamma t}$. Therefore by invoking the result from \cref{thm:hardnessForSpecificEll} we can say that for these values of $\ell, \ell'$ there is no bounded-error quantum algorithm that can solve $f_{\ell,\ell'}$ on the truth table of \CNF{} formulas in $\mathcal{O} \left( \left( \sqrt{\frac{N}{\ceil{\gamma t}}}+\frac{\sqrt{\ell (N-\ell)}}{\ceil{\gamma t}}\right)^{1-\delta}  \right)$ time, for each $\delta >0$; let us denote this claim by (*).

Let $\mathcal{A}$ be an algorithm that computes \appCountingCNFSAT{}, 
i.e., Algorithm $\mathcal{A}$ returns a value $h$ such that $(1-\gamma)\hat{h} < h < (1+\gamma)\hat{h}$. 
Given $\ell=\frac{N}{2}+\ceil{\gamma t}$ and $\ell'=\frac{N}{2}-\ceil{\gamma t}$, there are values of $t \in [N]$ such that we will be able to distinguish whether the number of satisfying assignments to a formula is $\ell$ or $\ell'$ using Algorithm $\mathcal{A}$. As $\ell >\ell'$ in our setup, we want $t$ such that $\ell'(1+\gamma) < \ell(1-\gamma)$; it is then necessary that $\gamma N < 2\ceil{\gamma t}$; let us denote this as Condition~1.

Now we set the values of $\ell$ and $\ell'$. Given a value of $\gamma \in \left[\frac{1}{N},\frac{1}{2}\right)$, we set $\ell=\left\lceil{\frac{N}{2(1-\gamma)}}\right\rceil$ and $\ell'=N-\ell$. This implies $\frac{N}{2(1-\gamma)} \leq \ell < \frac{N}{2(1-\gamma)}+1$,  $\frac{N(1-2\gamma)}{2(1-\gamma)}-1 < \ell' \leq \frac{N(1-2\gamma)}{2(1-\gamma)}$, and $\frac{\gamma N}{(1-\gamma)}\leq|\ell-\ell'|< \frac{\gamma N}{(1-\gamma)}+2$. Therefore we obtain $2\gamma \ell-2\gamma \leq|\ell-\ell'|  < 2\gamma \ell+2$.\footnote{To view the calculations in a less cumbersome way one can use the fact that asymptotically $\ell=\frac{N}{2(1-\gamma)}$, $\ell'=\frac{N(1-2\gamma)}{2(1-\gamma)}$ and $|\ell-\ell'| = \frac{\gamma N}{(1-\gamma)} = 2\gamma \ell$.} We know from claim (*) that every quantum algorithm that (for these values of $\ell, \ell'$) computes $f_{\ell,\ell'}$ on \CNF{} formulas requires $\Omega(L^{1-\delta})$ time for each $\delta>0$, where $L=\frac{1}{\gamma}\sqrt{\frac{N-\ell}{\ell+1}} = \Omega \left(\frac{1}{\gamma}\sqrt{\frac{N-\ell}{\ell}}\right)$. Moreover, $\ell'$ 
is $(\ell-1)(1-2\gamma)-1 < \ell' \leq \ell (1-2\gamma)$. Therefore, we can see that $L=\Omega \left( \frac{1}{\gamma}\sqrt{\frac{N-\ell}{\ell}} \right)=\Omega \left(\frac{1-2\gamma}{\gamma}\sqrt{\frac{N-\ell'}{\ell'}}\right)=\Omega \left(\frac{1}{\gamma}\sqrt{\frac{N-\ell'}{\ell'}}\right)$. 

It is also easy to see that if $\ell=\ceil{\frac{N}{2(1-\gamma)}}$ were to be expressed as $\frac{N}{2}+\ceil{\gamma t}$ (i.e. denote $\ell$ to be $\frac{N}{2}+\ceil{\gamma t}$), then for that value of $t$ we have $\ceil{\gamma t} = \ceil{\frac{N}{2(1-\gamma)}}- \frac{N}{2}\geq \frac{N\gamma}{2(1-\gamma)} > \frac{N \gamma}{2}$, which satisfies Condition~1. Hence here we can use Algorithm $\mathcal{A}$ to distinguish whether the number of satisfying assignments to a formula is $\ell$ or $\ell'$. Hence given a \CNF{} formula as input, we will be able to use Algorithm $\mathcal{A}$ to distinguish whether the number of satisfying assignments is $\ell$ or $\ell'$. Let $T= \frac{1}{\gamma}\sqrt{\frac{N-\ell}{\ell}}+ \frac{1}{\gamma}\sqrt{\frac{N-\ell'}{\ell'}} =\mathcal{O}(\frac{1}{\gamma}\sqrt{\frac{N-\ell'}{\ell'}})$. 
If  for some constant $\delta>0$, $\mathcal{A}$ can solve \appCountingCNFSAT{} on an input \CNF{} formula that has $\hat{h}$ number of satisfying assignments in $\mathcal{O}((\frac{1}{\gamma}\sqrt{\frac{N-\hat{h}}{\hat{h}}})^{1-\delta})$ time, then we are essentially computing $f_{\ell,\ell'}$ in $\mathcal{O}(T^{1-\delta})$ time, which is a contradiction to claim (*). Hence the first part of the statement of \cref{thm:AppCountQSETH} proved.

Proof of the second part of this theorem follows from \cref{thm:countingQSETH} as the regime $\gamma \hat{h} \leq 1$ translates to exactly counting the number of satisfying assignments.
\end{proof}

\subsection{Quantum complexity of $\#\kSAT{}$ and other related problems}
\label{sec:CountingKSAT}
In the previous subsection, we discussed the quantum complexity of variants of CNFSAT problems. However, it is not clear how to immediately derive a similar quantum complexity result for variants of $k$-SAT problems with constant $k$ 
by using the quantum (conditional) hardness results for variants of CNFSAT problems. Of course we could make a further conjecture about variants of $k$-SAT problems like we did in the previous subsection, but it would introduce too many conjectures. Moreover, some variants of $k$-SAT (for constant $k$) are even shown to be solvable in polynomial time~\cite{AW21}. 

To give the (quantum) complexity of some optimization problems (for example, lattice problems~\cite{BGS17}), on the other hand, we might want to have some (quantum) conditional lower bounds for (variants of) $k$-SAT problems with not too large $k$. This is because we might make $2^k\cdot \poly(n)$ calls to a solver of those problems to solve $k$-SAT. This is undesirable for giving the (quantum) complexity of those optimization problems when $k$ approaches $n$, while it is tolerable for a relatively small $k$ (like $k=\poly\log n$). 
Hence in this subsection, we would like to say something interesting about quantum hardness for \#$k$-SAT and $\oplus k$-SAT when $k=\Theta(\log n)$, only using the hardness assumptions on counting-CNFSAT (that is, \#QSETH). 
Here, variants of $k$-SAT are defined exactly the same way as \cref{def:PopularVariantsCNF}, \cref{def:AddErrorCountCNFSAT}, and \cref{def:ApproximateCountCNFSAT}, except that the input is now a $k$-CNF formula.


We use the classical algorithm by Schuler \cite{Schuler-AlgoCNF-2005}.\footnote{This algorithm can also be used to solve CNFSAT on $n$ variables, $m$ clauses in $\mathcal{O}(\poly(n)2^{n(1-1/(1+\log m))})$ expected time.} This algorithm can be viewed as a Turing reduction from SAT with bounded clause density to SAT with bounded clause width, which was analyzed in \cite{Calabro-DualityWidthClause-2006}. 
The time complexity of this algorithm is upper bounded by $\binom{m+n/k}{n/k}\cdot \poly(m,n)$, where $m$ is number of clauses.

\begin{algorithm}[h]
\caption{$\reduceWidth{(\psi)}$}
\begin{algorithmic}[1]
\Require CNF formula $\psi$
\If {$\psi$ has no clause of width $>k$}
\State output $\psi$
\Else
\State let $C'=\{l_1, \ldots, l_{k'}\}$ be a clause of $\psi$ of width $k'>k$
\State $C=\{l_1, \ldots, l_k\}$
\State $\psi_0 \leftarrow \psi-\{C'\} \cup \{C\}$
\State $\psi_1 \leftarrow \psi \land \lnot l_1 \land \lnot l_2 \land \cdots \land \lnot l_k$
\State $\psi_1 \leftarrow$ Remove variables corresponding to $l_1, \ldots, l_k$ from $\psi_1$ by setting $l_1=0, \ldots, l_k=0$
\State $\reduceWidth{(\psi_0)}$
\State $\reduceWidth{(\psi_1)}$
\EndIf
\end{algorithmic}
\label{algo:ReduceWidth}
\end{algorithm}

\cref{algo:ReduceWidth} takes as input a CNF formula of width greater than $k$, and then outputs a list of $k$-CNF formulas $\psi_i$ where the solutions of the input formula is the union
of the solutions of the output formulas, i.e., $\sol(\psi)=\cup_i \sol(\psi_i)$, where $\sol(\phi)$ denotes the set of satisfying assignments to a formula $\phi$. In fact, it is not hard to see that the count of the number of satisfying assignments also is preserved, i.e., $|\sol(\psi)|=\sum_i |\sol(\psi_i)|$. 

\begin{lemma}[Implicit from Section~3.2 in \cite{Calabro-DualityWidthClause-2006}]\label{lemma:witdhreduce}  \cref{algo:ReduceWidth} takes as input a CNF formula $\psi$ on $n$ input variables, with $m$ clauses, that is of width strictly greater than $k$ and outputs a number of $k$-CNF formulas $\psi_i$ each defined on at most $n$ input variables and at most $m$ clauses such that $|\sol(\psi)|=\sum_i |\sol(\psi_i)|$. 
\end{lemma}
\begin{proof}
Let $\psi=C'_1 \land C'_2 \land \cdots \land C'_m$ be the input CNF formula to \cref{algo:ReduceWidth}. The algorithm finds the first clause $C'_i$ that has width $k'>k$. Let $C'_i=(l_1 \lor l_2 \lor \dots \lor l_{k'})$ and $C_i=(l_1 \lor l_2 \lor \dots \lor l_{k})$. The algorithm then constructs two formulas $\psi_0=(\psi-\{C'_i\}) \cup \{C_i\}$ and $\psi_1=\psi \land \lnot l_1 \land \lnot l_2 \land \cdots \land \lnot l_k$. Then the algorithm recursively calls the subroutine on $\psi_0$ and $\psi_1$. We now claim the following. 

\begin{claim}
\label{claim:ReduceWidthPreservesSolutionCount}
$\sol(\psi_0) \cap \sol(\psi_1)=\emptyset$ and $\sol(\psi)=\sol(\psi_0) \cup \sol(\psi_1)$, i.e., $\sol(\psi)=\sol(\psi_0) \sqcup \sol(\psi_1)$.
\end{claim}
\begin{claimproof}
Let $x \in \sol(\psi_0)$. Then the clause $C_i(x)=(l_1(x) \lor \cdots \lor l_k(x))$ should evaluate to 1. Equivalently, $\lnot (l_1(x) \lor \cdots \lor l_k(x))=0$. Using De Morgan's laws we know $(\lnot l_1(x) \land \cdots \land \lnot l_k(x))=0$, which means that $\psi_1(x)=\psi(x)\land \lnot l_1(x) \land \cdots \land \lnot l_k(x) =0$. A similar argument can be used to show that if $x \in \sol(\psi_1)$, then $x \notin \sol(\psi_0)$. Therefore, $\sol(\psi_0) \cap \sol(\psi_1)=\emptyset$.

What remains to show is $\sol(\psi)=\sol(\psi_0) \cup \sol(\psi_1)$. 
\begin{itemize}
    \item If $x \in \sol(\psi_0)$ then $C_i(x)=1$, which implies $C'_i(x)=1$. Therefore $x \in \sol(\psi)$. If $x \in \sol(\psi_1)$ then $\psi_1(x)=1$, but $\psi_1(x)=\psi(x) \land \lnot l_1(x) \land \lnot l_2(x) \land \cdots \land \lnot l_k(x)$ which means $\psi(x)=1$ as well. Therefore, if $x \in \sol(\psi_0) \cup \sol(\psi_1)$ then $x \in \sol(\psi)$.
    \item If $x \in \sol(\psi)$ then $\psi(x)=1$, which means $C'_i(x)=1$. However, $C'_i(x)= (l_1(x) \lor \cdots \lor l_k(x)) \lor (l_{k+1}(x) \lor \cdots \lor  l_k'(x))$. This means either $(l_1(x) \lor \cdots \lor l_k(x))=C_i(x)=1$ or $(l_{k+1}(x) \lor \cdots \lor  l_k'(x))=1$ or both evaluate to 1. If $C_i(x)=1$ then $\psi_0(x)=1$, which means $x \in \sol(\psi_0)$. If $C_i(x)=0$ then $\psi_1(x)=\psi(x) \land (\lnot C_i(x))=1$, which means $x \in \sol(\psi_1)$.
\end{itemize}
Therefore, $\sol(\psi)=\sol(\psi_0) \sqcup \sol(\psi_1)$.
\end{claimproof}

Using \cref{claim:ReduceWidthPreservesSolutionCount} we conclude that $\sol(\psi)=\sqcup_i \sol(\psi_i)$, hence $|\sol(\psi)|=\sum_i |\sol(\psi_i)|$.
\end{proof}

Using \cref{lemma:witdhreduce} and Lemma 5 in \cite{Calabro-DualityWidthClause-2006} we will now show the hardness of $k$-SAT and its counting variants when $k=\Theta(\log n)$ without introducing new conjectures. 

\begin{corollary}\label{cor:countingksat}
For each constant $\delta>0$, there exists a constant $c$ such that there is no bounded-error quantum algorithm that solves
\begin{enumerate}
    \item $c\log n$-SAT in $\mathcal{O}(2^{(1-\delta)n/2})$ time unless \BasicQSETH{} (see \cref{cor:basicQSETH}) is false; 
    \item $\#$$c\log n$-SAT in $\mathcal{O}(2^{(1-\delta)n})$ time unless \#QSETH (see \cref{thm:countingQSETH}) is false;
    \item $\oplus$$c\log n$-SAT in $\mathcal{O}(2^{(1-\delta)n})$ time unless $\oplus$QSETH (see \cref{thm:ParityQSETH}) is false;
    \item $\oplus_q$$c\log n$-SAT in $\mathcal{O}(2^{(1-\delta)n})$ time unless $\#_q$QSETH (see \cref{conj:qParitySAT}) is false.
\end{enumerate}
\end{corollary}
\begin{proof}
We first prove the first item. Suppose that for each constant $c$, 
there is an algorithm $\mathcal{A}$ that solves \#$c\log n$-SAT in $2^{ns}$ for some constant $s:=1-\delta<1$. Let $k=c\log n$ for the rest of the proof. 
Consider the $\reduceWidth$ algorithm (\cref{algo:ReduceWidth}) with input CNF formula $\psi$. 
Let $p$ be some path of length $t$ in the tree $T$ of recursive calls to ReduceWidth$_k(\psi)$. Let $\psi_p$ be the output formula of width at most $k$ at the leaf of $p$. 
Let $l,r$ be the number of left, right branches respectively on path $p$. Every left branch in the path reduces the width of exactly 1 clause to $k$, therefore $l \leq m$. On the other hand, with additional $\poly(n,m)$ time, every right branch of path $p$ reduces the number of variables by $k$, therefore $r \leq n/k$. As a result, the number of paths in tree $T$ with $r$ right branches is at most $m+r \choose r$ and each outputs a formula with $n-rk$ variables.

Using the same arguments as in~\cite[Lemma~5]{Calabro-DualityWidthClause-2006}, one can see that $\mathcal{A}$ together with the $\reduceWidth$ subroutine can be used to solve \#CNFSAT (ignoring $\poly(n)$ factors) in time at most
\begin{align*}
     & \sum_{r=0}^{n/k} {m+r \choose r} 2^{s(n-rk)} +\binom{m+n/k}{n/k}\cdot \poly(m,n)\\
     \leq & \sum_{r=0}^{n/k} {m+\frac{n}{k} \choose r} 2^{s(n-rk)} \\
    = & 2^{sn} \sum_{r=0}^{n/k} {m+\frac{n}{k} \choose r} \frac{1}{2^{srk}} \\
    \leq & 2^{sn} (1 + \frac{1}{2^{sk}})^{m+\frac{n}{k}} \\
    \leq & 2^{sn} e^{\frac{1}{2^{sk}}(m+ \frac{n}{k})}  \hspace{5em} \text{since } (1+x) \leq e^x \\
    \leq & 2^{sn + \frac{4m}{2^{sk}}}, 
\end{align*}
where the last equality holds because we can assume that $m \geq \frac{n}{k}$ without loss of generality (by appending dummy clauses).
Therefore,  for each $c'$, 
there exist a constant $c$ for $k=c \log n$ and $\delta'$ such that if $m\leq c'n^2$, then $s + \frac{4m}{n2^{sk}}<1-\delta'$. As a result, a $2^{ns}$-time algorithm for  \#$c\log n$-SAT implies a $2^{n(1-\delta')}$-time 
algorithm for \#CNFSAT (restricted to formulas with $m\leq c'n^2$), 
which would refute \#-QSETH (\cref{thm:countingQSETH}). This proves the first item of the corollary. The same arguments hold for $\kSAT$, $\oplus\kSAT$, and $\oplus_q\kSAT$ as well.
\end{proof}

Note that, we \emph{cannot} extend the same arguments for the \propertyMajority{} or \propertyStrictMajority{} or additive-error approximation of count because those properties are not count-preserving. However, these arguments do extend to the multiplicative-factor approximation of the count.

\begin{corollary}
\label{cor:MultErrorkSAT}
Let $\gamma \in \left[\frac{1}{2^{n}},0.4999\right)$.  For each constant $\delta>0$, there exists constant $c$ such that, there is no bounded-error quantum algorithm that  solves $\gamma$-\#$c\log n$-SAT{} 
in time
\begin{enumerate}
    \item $\mathcal{O}\left( \left( \frac{1}{\gamma}\sqrt{\frac{2^{n}-\hat{h}}{\hat{h}}}\right)^{1-\delta}\right)\text{ if }\gamma \hat{h} > 1$, where $\hat{h}$ is the number of satisfying assignments;
    \item $\mathcal{O}(2^{n(1-\delta)})$ otherwise,
\end{enumerate} 
unless $\gamma$-\#QSETH (see \cref{thm:AppCountQSETH}, implied by \cref{thm:hardnessForSpecificEll}) is false.
\end{corollary}

\section{Quantum strong simulation of quantum circuits}
\label{sec:HardnessStrongSimulation}

We use the phrase \emph{strong simulation problem} to mean strong simulation of {quantum} circuits which is defined as follows.\footnote{Note that this is different from the \emph{weak simulation} problem; a weak simulation \emph{samples} from probability distribution $p(x):=|\bra{0^n}C\ket{x}|^2$.}

\begin{defn}[The strong simulation problem] 
\label{def:StrongSimulation} Let $p\in \mathbb{N}$.
Given a quantum circuit $C$ on $n$ qubits and $x \in \{0,1\}^n$, the goal of \emph{strong simulation} with $p$-bit precision is to output the value of $|\bra{x}C\ket{0^n}|:=0.C_1C_2\ldots$ up to $p$-bit precision. That is, output $C_0.C_1\ldots C_{p-1}$.\footnote{Though in some papers the strong simulation problem requires that we output $\bra{x}C\ket{0^n}$ instead of $|\bra{x}C\ket{0^n}|$, we use this definition because it is more comparable to the definition of the weak simulation problem. Also, the lower bound we present holds for both of these definitions.}
\end{defn}

For a quantum circuit $C$, computing $|\bra{x}C\ket{0^n}|$ exactly, to a precision of $n$ bits, is \#P-hard \cite{CHM21, VDN10}. This means even a \emph{quantum} computer will likely require exponential time to strongly simulate another quantum circuit. In this section, we prove a more \emph{precise} quantum time bound for strongly simulating quantum circuits, both exactly and approximately; in the approximate case, we present complexity results for both multiplicative factor and additive error approximation. Our results extend the results by \cite{HNS20} in two directions: firstly, we give explicit (conditional) bounds proving that it is hard to strongly simulate quantum circuits using \emph{quantum} computers as well. Secondly, we also address the open question posed by \cite{HNS20} on the (conditional) hardness of strong simulation up to accuracy $O(2^{-n/2})$, however, our results are based on a hardness assumption different from SETH or Basic QSETH.

The results presented in this section are based on two main components. Firstly, on the observation that the reduction from \CNFSAT{} to the strong simulation problem given (in Theorem~3) by \cite{HNS20} encodes the count of the number of satisfying assignments; this fact allows us to use the same reduction to reduce other variants of \CNFSAT, such as \CountingCNFSAT{} or \ParityCNFSAT{}, to the strong simulation problem, moreover, the same reduction also allows us to reduce \appCountingCNFSAT{} and \addErrorCountingCNFSAT{} to analogous variants of the strong simulation problem, respectively. As the second main component, we use the quantum hardness of these variants of \CNFSAT{} problem discussed in \cref{sec:LowerBoundsVariants}.

We will first state the result of the exact quantum time complexity of the strong simulation problem and then use that result to later show how hard it is for a quantum computer to strongly simulate a quantum circuit with an additive error or a multiplicative factor approximation.

\begin{figure}
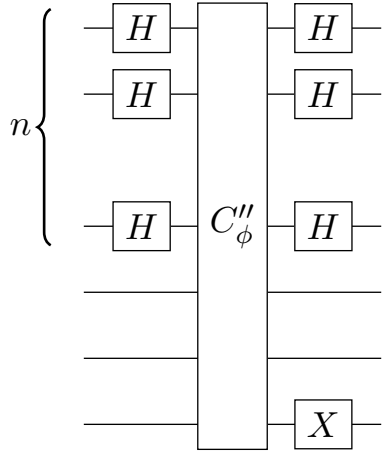

    \centering
    \include{CircuitStrongSimulationCircuit}
    \caption{The circuit $C_{\phi}$.}
    \label{fig:CircuitEncodingTotalSatAssignments}
\end{figure}

\begin{theorem}
\label{thm:QHardnessExactStrongSimulation}
For each constant $\delta>0$, there is no bounded-error quantum algorithm that solves the strong simulation problem (\cref{def:StrongSimulation}) up to a precision of $n+1$ bits 
in $\mathcal{O}(2^{n(1-\delta)})$ time, unless \#QSETH (see \cref{thm:countingQSETH}) is false. 
\end{theorem}
The proof is similar to the proof of \cite[Theorem~3]{HNS20}. We restate it here for ease of reading.

\begin{proof}
Let $\phi$ be a \CNF{} formula on $n$ input variables and $m$ clauses. Let $C'_{\phi}$ denote the \emph{reversible} classical circuit, using TOFFOLI, CNOT and X gates, that \emph{tidily} computes $\phi(x)$ for all $x \in \{0,1\}^n$.\footnote{A classical circuit $C:\{0,1\}^{n+w(n)+1} \rightarrow \{0,1\}^{n+w(n)+1}$ \emph{reversibly and tidily computes} a function $f:\{0,1\}^n \rightarrow \{0,1\}$ if the following statements are true.
\begin{enumerate}
    \item Circuit $C$ \emph{reversibly} computes $f$ if $C$ consists of reversible gates, such as $\{\text{TOFFOLI, CNOT, NOT}\}$, and 
    \begin{equation}
        \forall x \in \{0,1\}^n, \exists W(x) \in \{0,1\}^{w(n)}, C(x,0^{w(n)}, b)=(x,W(x), b\oplus f(x)).
    \end{equation}
    \item Circuit $C$ \emph{tidily} computes $f$ if
    \begin{equation}
        \forall x \in \{0,1\}^n, \forall b \in \{0,1\}, C(x,0^{w(n)},b)=(x,0^{w(n)},b\oplus f(x)).
    \end{equation}
\end{enumerate}} One can construct circuit $C'_{\phi}$ of width $k$ with $n \leq k \leq n+2(\ceil{\log n} + \ceil{\log m})$ and size $s \leq 8 \times 3^{\ceil{\log n}+\ceil{\log m}}$ with sublinear space and polynomial-time overhead;  see \cite[Section~4.1]{HNS20}. Let $C''_{\phi}$ denote the quantum analogue of the classical reversible circuit $C'_{\phi}$, i.e., for the gates in $C'_{\phi}$, TOFFOLI and CNOT remain unchanged. Therefore, the width and size of $C''_{\phi}$ remains $k$ and $s$, respectively.

Clearly, $C''_{\phi}$ maps $ \mathcal{H}$ to $  \mathcal{H}$ where $\mathcal{H}$ denotes a $2^k$-dimensional Hilbert space. Then, we see that $C''_{\phi} \ket{x} \ket{0^{k-n}}=\ket{x}\ket{\phi(x)}\ket{0^{k-1-n}}$. Let $C_{\phi}$ denote the quantum circuit in \cref{fig:CircuitEncodingTotalSatAssignments}. Width of this circuit is still $k$ and size is $\mathcal{O}(s)$. One can see that $\bra{0^k}C_{\phi}\ket{0^k}$ encodes the fraction of satisfying assignments to formula $\phi$ (that is, $|\phi|/2^n$). 

If there exists a constant $\delta>0$ such that {strong quantum simulation} of circuit $C_{\phi}$ on basis state $\ket{x}=\ket{ 0^k}$ can be computed in $T=2^{k(1-\delta)}$ time up to $n+1$ bits of precision, then this could count the number of satisfying assignments formula of $\phi$ in time $T$ exactly. Plugging in the values of $k \leq n+2(\ceil{\log n}+\ceil{\log m})$ and $s \leq 8 \times 3^{\ceil{\log n}+\ceil{\log m}}=\poly(n,m) $ we get $T \leq \mathcal{O}(2^{n(1-\delta)})\cdot \poly(m)$ time. This would refute \#QSETH (\cref{thm:countingQSETH}). 
\end{proof}

Note that even if we only care about the first two bits $C_0.C_1$ of $|\bra{x}C\ket{0^n}|=C_0.C_1C_2\ldots$, it is still hard to determine which values $C_0$ and $C_1$ are, because it means we determine if the number of satisfying assignments is $\geq 2^n/2$ or not (if $C_0C_1=10$ or $01$, then $|\phi|\geq 2^n/2$, and if $C_0C_1=00$, then $|\phi|< 2^n/2$). Therefore, using exactly the same statement in the proof above, we obtain the following corollary that gives the same lower bound for $2$-bit precision, by using a different hardness assumption (which is Majority-QSETH). 

\begin{corollary}
\label{thm:StrongSimulationExactTwoBits}
For each constant $\delta>0$, there is no bounded-error quantum algorithm that solves the strong simulation problem (\cref{def:StrongSimulation}) up to $2$ bits of precision 
in $\mathcal{O}(2^{n(1-\delta)})$ time, unless Majority-QSETH (Part 1 of \cref{thm:MajorityQSETH}) is false. 
\end{corollary}

One can also try to solve the strong simulation problem with additive-error approximation using the following definition.

\begin{defn}(Strong simulation with additive error $\Delta'$)
\label{def:AdditiveErrorStrongSimulation} Let $\Delta' \in [\frac{1}{2^{n+1}},1)$.
Given a quantum circuit $C$ on $n$ qubits and $x \in \{0,1\}^n$, the goal of strong simulation with additive error $\Delta'$ is to estimate $d'=|\bra{0^n}C\ket{x}|$ with additive error $\Delta'$.    
\end{defn}
Using the same reduction in the proof of \cref{thm:QHardnessExactStrongSimulation} and the conjectured hardness of the \addErrorCountingCNFSAT{} problem we can immediately get the following corollary.\footnote{Note that the value of $k$ in relation to $n$ is such that $2^k = \widetilde{\mathcal{O}}(2^n)$; here $k$ refers to the $k$ used in the proof of \Cref{thm:QHardnessExactStrongSimulation} instead of the $k$ of $k$-SAT.}

\begin{corollary}\label{cor:DeltaErrorQSim} For each constant $\delta>0$, there is no bounded-error quantum algorithm that solves strong simulation
with additive error $\Delta'=\frac{\Delta}{2^n} \in [\frac{1}{2^{n}},1)$ in  $\widetilde{\mathcal{O}}\left(\left( \sqrt{\frac{2^n}{\Delta}} + \frac{\sqrt{\hat{h} (2^n -\hat{h})}}{\Delta}\right)^{1-\delta}\right)$ time 
where $\hat{h}=\bra{0^n} C \ket{x}\cdot 2^n$, unless $\Delta$\textsc{-add-}\#QSETH (see \cref{thm:AddErrorQSETH}) is false.
\end{corollary}

It is beneficial to note that we only get (at best) a $\poly(n)$-time quantum lower bound for the strong simulation problem when $\Delta'=\frac{\Delta}{2^n}=\Theta(1)$. Fortunately, this lower bound matches the $\poly(n)$ time quantum upper bound for the strong simulation problem when $\Delta'=\Theta(1)$, see the end of this section (\Cref{thm:upper_qstrongsim}) for details. In fact, for some values of $\hat{h}$ our lower bounds are also tight in terms of $\Delta'$. 
Additionally, we can use a similar argument to show strong simulation results with multiplicative factor, defined as follows.

\begin{defn}(Strong simulation with multiplicative factor $\gamma$)
\label{def:MultiplicativeErrorStrongSimulation}
Let $\gamma>0$. 
Given a quantum circuit $C$ on $n$ qubits and $x\in\{0,1\}^n$, the goal of strong simulation with multiplicative factor $\gamma$ is to estimate the value $d'=|\bra{0^n}C\ket{x}|$ with  multiplicative error $\gamma$, i.e., output a value $d$ such that $(1-\gamma)d' < d < (1+\gamma)d'$. 
\end{defn}

The exact arguments in the proof of \cref{thm:QHardnessExactStrongSimulation} can be used to prove the following statement.

\begin{theorem}\label{thm:gammaQsim}
Let $\gamma \in [\frac{1}{2^n},0.4999)$. For each constant $\delta>0$, there is no bounded-error quantum algorithm that can solve the strong simulation problem with multiplicative error $\gamma$ 
in time
\begin{enumerate}
    \item $\mathcal{O}\left( \left( \frac{1}{\gamma}\sqrt{\frac{2^{n}-\hat{h}}{\hat{h}}}\right)^{1-\delta}\right)$ if $\gamma \hat{h} \geq 1$, where $\hat{h}=\bra{0^n}C \ket{x}\cdot 2^n$ for input $x \in \{0,1\}^n$;
    \item $\mathcal{O}(2^{n(1-\delta)})$, otherwise,
\end{enumerate}
    unless $\gamma$-\#QSETH (see \cref{thm:AppCountQSETH}, implied by \cref{thm:hardnessForSpecificEll}) is false.
\end{theorem}
\begin{proof}
Let $\mathcal{A}$ denote an algorithm that for a given $s$-sized quantum circuit $C$ on $k$ qubits and a given $x \in \{0,1\}^k$, for some constants $\delta$ and $\delta'$, computes $\bra{0^k} C \ket{x}$ with  multiplicative error $\gamma$ either 
\begin{enumerate}
    \item in $\mathcal{O}\left( \left( \frac{1}{\gamma}\sqrt{\frac{2^{k}-\hat{h}}{\hat{h}}}\right)^{1-\delta}\right)$ time, whenever $\gamma \hat{h} \geq 1$, or
    \item in $\mathcal{O}(2^{k(1-\delta')})$ time, otherwise.
\end{enumerate}

Then, given a \CNF{} formula $\phi$ on $n$ input variables and $m$ clauses, one can do the following to approximately count the number of satisfying assignments to $\phi$: first use the $\poly(n,m)$ reduction as in the proof of \cref{thm:QHardnessExactStrongSimulation} to construct the quantum circuit $C_{\phi}$ of size $s=\poly(n,m)$ and width $k\leq n+2(\ceil{\log n}+\ceil{\log m})$. Then run algorithm $\mathcal{A}$ on quantum circuit $C_{\phi}$ and $x=0^k$ as inputs. The output would then be a $\gamma$-multiplicative approximation of the count of the number of satisfying assignments to $\phi$. Depending on the value of $\gamma \cdot \hat{h}$, the running time of this entire process is either $\mathcal{O}\left( \poly(n,m)+ \left( \frac{1}{\gamma}\sqrt{\frac{2^{k}-\hat{h}}{\hat{h}}}\right)^{1-\delta} \right)=\mathcal{O}\left( \left( \frac{1}{\gamma}\sqrt{\frac{2^{n}-\hat{h}}{\hat{h}}}\right)^{1-\delta} \cdot \poly(m) \right)$ for some constant $\delta>0$, or 
$\mathcal{O}(2^{k(1-\delta')})=\mathcal{O}(2^{n(1-\delta')}\cdot \poly(m))$ time for some constant $\delta'>0$. Either way refutes $\gamma$-\#QSETH (\cref{thm:AppCountQSETH}).
\end{proof}

\paragraph{Quantum upper bounds for strong simulation}
\label{Appsec:UpperBoundAddStrongSimulationProblem}

Here, we include a quantum algorithm for strong simulation with additive error $\Delta'$ for completeness. 

    


\begin{theorem}\label{thm:upper_qstrongsim}
Let $\Delta' \in [\frac{1}{2^{n+1}},1)$. There exists a quantum algorithm that solves strong simulation with additive error $\Delta'$ (\cref{def:AdditiveErrorStrongSimulation}) in $\poly(n,|C|)\cdot\frac{1}{\Delta'}$ time, where $|C|$ is the size (the number of quantum elementary gates it contains) of input circuit $C$.
\end{theorem}
\begin{proof}
Given a quantum circuit $C$ on $n$ input variables and $x \in \{0,1\}^n$, the task (of strong simulation with additive error $\Delta'$) is to estimate the value of $|\bra{x}C\ket{0^n}|$ with $\Delta'\in [\frac{1}{2^{n+1}},1)$ additive-error approximation. 

Let $\ket{\psi}= C\ket{0^n}:=\sum_{i \in \{0,1\}^n} \alpha_i \ket{i}$. 
Let $U'$ denote the unitary $U':\ket{i}\ket{x}\ket{b} \rightarrow \ket{i}\ket{x}\ket{b\oplus (i=x)}$ for $i,x \in \{0,1\}^n$ and $b \in \{0,1\}$. It is easy to verify that combining $C$ and $U'$ we can construct a unitary $U$ on $2n+1$ qubits such that 
\begin{equation*}
    U\ket{0^{2n}}\ket{0} = \alpha_x\ket{x}\ket{x}\ket{1}+\sum_{i\neq x}\alpha_i \ket{i}\ket{x}\ket{0},
\end{equation*}
and $|\alpha_x|=|\bra{x}C\ket{0^n}|$. Using the amplitude estimation algorithm (\cref{thm:AmplitudeEstimation}), we can estimate $|\bra{x}C\ket{0^n}|$ to an additive error of $\Delta'$ in $\poly(n,|C|)\cdot\frac{1}{\Delta'}$ time.
\end{proof}

\section{Quantum lower bound for lattice counting and q-count problems}\label{sec:lattice}

In this section, we would like to connect $k$-SAT to variants of lattice problems and then use the QSETH lower bound we have in \cref{sec:CountingKSAT} to give quantum fine-grained complexity for those lattice problems. 

Fine-grained complexity of lattice problems is quite widely studied in the classical case~\cite{BGS17,DBLP:conf/soda/AggarwalBGS21,AggarwalC21,AggarwalS18,BennettP20,BennettPT22}. Lots of variants of lattice problems have been considered before, and the most well-studied one is the closest vector problem (with respect to $\ell_p$ norm).

CVP$_p$ is known to have a $2^n$ SETH lower bound for any $p\not\in 2\intg$~\cite{BGS17,DBLP:conf/soda/AggarwalBGS21}, and for even $p$, there seems a barrier for showing a fine-grained reduction from $k$-SAT to CVP~\cite{AK22CVPimpossible}. Kannan gave a $n^{\mathcal{O}(n)}$-time algorithm for solving CVP$_p$ for arbitrary $p\geq 1$~\cite{Kan83IP}, while the best-known algorithm for solving CVP$_p$ with noneven $p$ is still $n^{cn}$ for some constant $c$. To get a conditional quantum lower bound for CVP$_p$ for noneven $p$, given there is already a classical reduction from $k$-SAT to CVP$_p$ using $2^k\cdot \poly(n)$ time (for noneven $p$)~\cite{BGS17,DBLP:conf/soda/AggarwalBGS21}, either one can directly use the QSETH framework by Aaronson et al.~\cite{BasicQSETH-Aaronson-2020} to get a $2^{(0.5-\epsilon)n}$ lower bound, or we can use \cref{cor:countingksat} to get the same lower bound in our QSETH framework.\footnote{Basic QSETH assumption is weaker than the QSETH assumption in Aaronson et al~\cite{BasicQSETH-Aaronson-2020}, so our lower bound under basic QSETH assumption~(\cref{conj:ACqseth,cor:basicQSETH}) will also imply a lower bound under their quantum SETH framework.}

A natural question is invoked here: Can we have a $2^{(0.5+\epsilon)n}$ quantum SETH lower bound for any (variants of) lattice problems? The answer is yes by using the framework and the problems introduced in \cref{sec:CountingKSAT} and by considering the counting variant of lattice problems.  We begin by introducing the (approximate) lattice counting problem and some other related problems as follows:

\begin{defn}[Lattice counting problem]
Let $\gamma\geq 0$ and $1\leq p \leq \infty$. The $\gamma$-approximate Vector Counting Problem $\gamma$-VCP$_p$ is the counting problem defined as follows: The input is a basis $\mathbf{B}\in \mathbb{R}^{d\times n}$ of a lattice $\cL(\mathbf{B})$, target vector $\vect{t}\in \mathbb{R}^{d}$, and radius $r\in \mathbb{R}_+$. The goal of this problem is to output a value $C$ satisfying $|(\cL-\vect{t})\cap r\cdot B^{d}_p|\leq C\leq (1+\gamma)\cdot|(\cL-\vect{t})\cap r\cdot B^{d}_p|$. If $\gamma=0$, we simply denote the problem as VCP$_p$.
\end{defn}

The (approximate) lattice counting problem was first introduced by Stephens-Davidowitz as a promise problem~\cite{SteDiscreteGaussian16}, and here we slightly modify the definition to make it a counting problem. We also generalize the lattice vector counting problem to the $q$-count problem, as follows.

\begin{defn}[Lattice $q$-count problem]
Let $\gamma\geq 1$, $1\leq p \leq \infty$, and $q\in [2^n]\setminus \{1\}$. The lattice $q$-count Problem $\#_q$-VCP$_p$ is the lattice $q$-count problem defined as follows: The input is a basis $\mathbf{B}\in \mathbb{R}^{d\times n}$ of a lattice $\cL(\mathbf{B})$, target vector $\vect{t}\in \mathbb{R}^{d}$, and radius $r\in \mathbb{R}_+$. The goal of this problem is to output a value $C=|(\cL-\vect{t})\cap r\cdot B^{d}_p|\mod q$. If $q=2$, then we simply denote the problem as $\oplus$VCP$_p$.
\end{defn}

One can consider the above two problems as the counting and the q-counting version of CVP$_p$, respectively. To connect these problems to the counting $k$-SAT problem, we should first introduce the following geometric tool introduced by Bennett, Golovnev, and Stephens-Davidowitz~\cite{BGS17}.

\begin{defn}[Isolating parallelepiped]
Let $k$ be an integer between $3$ and $n$ and $1\leq p \leq \infty$. We say that $V \in \mathbb{R}^{2^k\times k}$ and $\vect{u}\in \mathbb{R}^{2^n}$ define a $(p, k)$-isolating parallelepiped if $\|\vect{u}\|_p> 1$ and $\|V\vect{x}-\vect{u}\|_p$ = 1 for all $\vect{x}\in \{0,1\}^k\setminus \{0_k\}$.
\end{defn}

For the sake of completeness, we will explain how to connect a $k$-CNF formula to an instance of VCP using the above object. The proof is very similar to the proof of~\cite[Theorem~3.2]{BGS17}. 
 
\begin{theorem}\label{thm:CNFtoVCP}
Let $k$ be an integer between $3$ and $n$. Suppose we have a $(p, k)$-isolating parallelepiped $(V,\vect{u})$ for some $p = p(n) \in [1, \infty)$ and can make quantum queries to oracles $O_V:\ket{i}\ket{s}\ket{0}\rightarrow \ket{i}\ket{s}\ket{V_{is}}$ and $O_{\vect{u}}: \ket{i}\ket{0}\rightarrow\ket{i}\ket{\vect{u}_i}$ for $i\in [2^k]$ and $s\in [k]$. Then for every given input oracle $O_\Phi:\ket{j}\ket{w}\ket{0}\rightarrow \ket{j}\ket{k}\ket{C_w(j)}$ of k-CNF formula (with $m$ clauses) $\psi=C_1\land C_2 \land \cdots \land C_m$ for $w\in [m]$ and $j\in[n]$, one can output oracles $O_\mathbf{B}:\ket{h}\ket{j}\ket{0}\rightarrow \ket{h}\ket{j}\ket{\mathbf{B}_j(h)}$ of basis $\mathbf{B}\in \mathbb{R}^{(m\cdot 2^k+n)\times n}$ for each $h\in[m\cdot 2^k+n]$ and $j\in[n]$, $O_{\vect{t}}:\ket{h}\ket{0}\rightarrow\ket{h}\ket{t_h}$ of target vector $\vect{t}\in \mathbb{R}^{m\cdot 2^k+n}$ for each $h\in[m\cdot 2^k+n]$, and radius $r$ such that $|$sol($\psi)|=$VCP$_p(\mathbf{B},\vect{t},r)$, using $\poly(n,m)$ queries to $O_V$, $O^\dagger_V$, $O_{\vect{u}}$, $O^\dagger_{\vect{u}}$, $O_\Phi$, $O_\Phi^\dagger$, and elementary gates.
\end{theorem}

\begin{proof}
Let $d = m\cdot 2^k+n$ and $V=\{\vect{v_1},\ldots,\vect{v_k}\}$ with $\vect{v_s}\in \mathbb{R}^{2^k}$ for every $s\in[k]$. The basis $\mathbf{B} \in \mathbb{R}^{d\times n}$ and target vector $\vect{t} \in \mathbb{R}^d$ in the output instance have the form:
\begin{align*}
   \mathbf{B}= \begin{pmatrix}
\mathbf{B}_1\\[\jot]
\vdots\\[\jot]
\mathbf{B}_m\\[\jot]
2\cdot m^{1/p}\cdot I_n\\[\jot]
\end{pmatrix},\hspace{5mm} \vect{t}= \begin{pmatrix}
\vect{t}_1\\[\jot]
\vdots\\[\jot]
\vect{t}_m\\[\jot]
m^{1/p}\cdot \mathbf{1}_n\\[\jot]
\end{pmatrix},
\end{align*}
with blocks $\mathbf{B}_w \in \mathbb{R}^{2^k\times n}$ that correspond to the clause $C_w = \vee_{s=1}^k \ell_{w,s}$ and $\vect{t}_w \in \mathbb{R}^{2^n}$ for each $w\in[m]$. For each $w\in[m]$ and $j\in [n]$, the $j$th column $(\mathbf{B}_w)_j$ of block $\mathbf{B}_w$ is
\begin{align*}
    (\mathbf{B}_w)_j=\begin{cases}
\vect{v_w}, \text{        if $x_j$ is the $w$th literal of clause $w$},\\
-\vect{v_w}, \text{      if $\neg x_j$ is the $w$th literal of clause $w$,}\\
0_{2^d}, \text{     otherwise,}
\end{cases}
\end{align*}
and $\vect{t}_w=\vect{u}-\sum\limits_{s\in N_w}\vect{v}_s$, where $N_w= \{s \in [k] : \ell_{w,s} \text{ is negative}\}$ is the set of the indices of negative
literals in $C_w$. Also set $r=(mn+m)^{1/p}$.

Define the unitary $U_e:\ket{a}\ket{b}\ket{0}\rightarrow \ket{a}\ket{b}\ket{\delta_{|a|,|b|}}$ for each $a,b \in [d]$, and the unitary as for each $x_1,\ldots,x_n\in \{0,1\},$
\[
U_{loc}: \ket{x_1}\ket{x_2}\ldots\ket{x_n} \ket{0} \rightarrow 
    \begin{dcases}
    \ket{x_1}\ket{x_2}\ldots\ket{x_n} \ket{s}, & \text{if only $x_s=1$ and all others are } 0,\\
    \ket{x_1}\ket{x_2}\ldots\ket{x_n} \ket{0}, &\text{otherwise,}
\end{dcases}
\]
which can both be implemented via $\poly(n)$ elementary gates up to negligible error. Also, we define the unitary as follows: for each $w\in[m]$, $j\in [n]$,
\[
U_{s\ell}: \ket{w}\ket{j}\ket{0} \rightarrow 
    \begin{dcases}
    \ket{w}\ket{j}\ket{s}, & \text{if $x_j$ is the $s$th literal of clause $C_w$} ,\\
    \ket{w}\ket{j}\ket{-s}, & \text{if $\neg x_j$ is the $s$th literal of clause $C_w$} ,\\
    \ket{w}\ket{j}\ket{0}, &\text{otherwise,}
\end{dcases}
\]
which can also be implemented by using $\poly(m,n)$ applications of $O_\Phi$ and $\poly(m,n)$ many elementary gates.

We are now ready to show how to implement $O_\mathbf{B}$ and $O_{\vect{t}}$. For input $\ket{h}\ket{j}\ket{0}\ket{0}\ket{0}\ket{0}$, let $\ket{h}=\ket{h_1}\ket{h_2}$ where $h_2$ is the last $n$ bits of $h$ and $h_1$ is the remaining prefix. Then we first apply $U_e$ to the first, third, and fifth registers to obtain $\ket{h_1}\ket{h_2}\ket{j}\ket{0}\ket{\delta_{h_1,0}}\ket{0}\ket{0}$. After that, apply $U_{s\ell}$ to the first, third, and sixth registers, and apply $O_V$ to the second, third, and seventh registers we get
$$
\ket{h_1}\ket{h_2}\ket{j}\ket{0}\ket{\delta_{h_1,0}}\ket{s}\ket{V_{js}},
$$
where $V_{js}=(\vect{v}_s)_j$. Finally, adding another ancilla register, we can store the value $V_{js}\cdot \delta_{h_1,0}+2m^{1/p}(1-\delta_{h_1,0})\delta_{h-m\cdot2^k,j}$ in the last register. Uncomputing the fourth to seventh registers, we have $$
\ket{h_1}\ket{h_2}\ket{j}\ket{V_{js}\cdot \delta_{h_1,0}+2m^{1/p}(1-\delta_{h_1,0})\delta_{h-m\cdot2^k,j}}=\ket{h}\ket{j}\ket{(\vect{v}_s)_j\cdot \delta_{h_1,0}+2m^{1/p}(1-\delta_{h_1,0})\delta_{h-m\cdot2^k,j}},
$$
and $(\vect{v}_s)_j\cdot \delta_{h_1,0}+2m^{1/p}(1-\delta_{h_1,0})\delta_{h-m\cdot2^k,j}$ is exactly the coefficient of $\mathbf{B}_j(h)$. One can see we only use $\poly(n,m)$ queries to $O_V$, $O^\dagger_V$, $O_\Phi$, $O_\Phi^\dagger$ and a similar number of elementary gates. We can also construct $O_{\vect{t}}$ using a similar strategy, which can also be done using at most $\poly(n,m)$ queries to $O_V$, $O^\dagger_V$, $O_\Phi$, $O_\Phi^\dagger$, $O_u$, $O^\dagger_u$, and elementary gates.

To see the correctness, 
consider $\vect{y}\in \mathbb{Z}^n$. If $\vect{y}\notin \{0,1\}^n$, then
$$
\|\mathbf{B}\vect{y}-\vect{t}\|^p_p \geq \|2m^{1/p}I_n\vect{y}-m^{1/p}\mathbf{1}_n\|^p_p\geq m(n+2).
$$
On the other hand, if $\vect{y}\in \{0,1\}^n$, then for each $\mathbf{B}_w$
\begin{align*}
    \|\mathbf{B}_w\vect{y}-\vect{t}\|^p_p&= \|\sum\limits_{s\in P_w}\vect{y}_{\text{ind}({\ell_{w,s}})}\cdot \vect{v}_s-\sum\limits_{s\in N_w}\vect{y}_{\text{ind}({\ell_{w,s}})}\cdot \vect{v}_s-(u-\sum\limits_{s\in N_w}\vect{y}_{\text{ind}({\ell_{w,s}})}\cdot \vect{v}_s)\|^p_p\\
    &= \|\sum\limits_{s\in P_w}\vect{y}_{\text{ind}({\ell_{w,s}})}\cdot \vect{v}_s+\sum\limits_{s\in N_w}(1-\vect{y}_{\text{ind}({\ell_{w,s}})})\cdot \vect{v}_s-u\|^p_p\\
    &= \|\sum\limits_{s\in S_w(\vect{y})}\vect{v}_s-u\|^p_p,
\end{align*}
where ${{\text{ind}({\ell_{w,s}})}}$ is the index of
the variable underlying $\ell_{w,s}$ (that is, ${\text{ind}({\ell_{w,s}})}=j$ if $\ell_{w,s}=x_j$ or $\neg x_j$), $P_w= \{s \in [k] : \ell_{w,s} \text{ is positive}\}$ is the set of the indices of positive literals in $C_w$, and $S_w(\vect{y})= \{s\in P_w:\vect{y}_{\text{ind}(\ell_{w,s})}=1\}\cup \{s\in N_w:\vect{y}_{\text{ind}(\ell_{w,s})}=0\}$ is the
indices of literals in $C_w$ satisfied by $\vect{y}$.
 Because $(V,\vect{u})$ is a $(p,k)$-isolating parallelepiped, if $|S_w(\vect{y})|\neq 0$, then $\|\sum\limits_{s\in S_w(\vect{y})}\vect{v}_s-\vect{u}\|^p_p=1$, and it will be greater than $1$ otherwise. Also, $|S_w(\vect{y})|\neq 0$ if and only if $C_w$ is satisfied. Therefore, we have that for every satisfying assignment $\vect{y}$, 
$$
\|\mathbf{B}\vect{y}-\vect{t}\|^p_p=\sum\limits_{k=1}^m \|\mathbf{B}_k\vect{y}-\vect{t}\|^p_p+mn=m+mn,
$$
and for all other unsatisfying assignments $\vect{y}'$, $\|\mathbf{B}\vect{y'}-\vect{t}\|^p_p>m+mn$. As a result, every satisfying assignment of $\psi=C_1\land C_2 \land \cdots \land C_m$ is encoded as a lattice point of lattice $\cL(\mathbf{B})$ with distance $r=(mn+m)^{1/p}$ to the target vector $t$, which implies that the answer of counting-SAT with input $\psi$ is equal to VCP$_p(\mathbf{B},\vect{t},r)$.
\end{proof}



The theorem above shows how to connect the counting $k$-SAT problem to the vector counting problem given access to a $(p,k)$-isolating parallelepiped. However, it is not always the case that we can compute such an isolating parallelepiped efficiently. 
Aggarwal, Bennett, Golovnev, and Stephens-Davidowitz~\cite{BGS17,DBLP:conf/soda/AggarwalBGS21} showed the existence of isolating parallelepiped for some $p,k$ and provided an efficient algorithm for computing them.

\begin{theorem}[\cite{DBLP:conf/soda/AggarwalBGS21}]\label{thm:ppcomp}
    For $k \in \mathbb{Z}_+$ and computable $p=p(n) \in [1,\infty)$ if $p$ satisfies either (1) $p \notin 2\mathbb{Z}$ or (2) $p\geq k$, the there exists a $(p, k)$-isolating parallelepiped $V \in\mathbb{R}^{2^k\times k}$, $\vect{u}\in \mathbb{R}^{2^k} $and it is computable in time $\poly(2^k,n)$.
\end{theorem}
Therefore, by choosing $k=\Theta(\log n)$ and combining \cref{cor:countingksat} and \cref{thm:CNFtoVCP,thm:ppcomp}, we can directly show a $2^n$-QSETH lower bound for VCP$_p$ for all non-even $p$. Also, a similar idea works for the $q$-count and approximate counting of CNF-SAT: for each CNF-SAT formula $\psi$, using \cref{algo:ReduceWidth} and \cref{lemma:witdhreduce}, we can output a number of $k$-CNF formulas $\psi_1\ldots,\psi_N$ such that $|\sol(\psi)|=\sum_{i\in[N]} |\sol(\psi_i)|$. Once we have an algorithm that solves $\gamma\#$-$k$-SAT ($\#_qk$-SAT), we can use it to compute $\gamma$-$\#k$-SAT$(\psi_i)$ ($\#_qk$-SAT$(\psi_i)$) for all $i\in[N]$, and then by adding the outputs together, we can get a valid solution to $\gamma$-$\#$CNFSAT ($\#_q$CNFSAT) with input $\psi$. 
Combining the above arguments with \cref{thm:CNFtoVCP} (and the proof of \cref{cor:countingksat}), we have the following corollaries.

\begin{corollary}\label{cor:latticecountQSETH}\label{cor:qcountlatticeQSETH}\label{cor:paritylatticeQSETH}
Let $p\in[1,\infty)\setminus 2\mathbb{Z}$ and $q\in [2^n]\setminus \{1,2\}$. For each constant $\delta>0$, there is no bounded-error quantum algorithm that solves
\begin{enumerate}
        \item VCP$_p$ in $\mathcal{O}(2^{n(1-\delta)})$ time, unless \#QSETH(see \cref{thm:countingQSETH}) is false;
        \item $\oplus$VCP$_p$ in $\mathcal{O}(2^{n(1-\delta)})$ time, unless $\oplus$QSETH(see \cref{thm:ParityQSETH})  is false;
        \item $\#_q$VCP$_p$ in $\mathcal{O}(2^{n(1-\delta)})$ time,  unless $\#_q$QSETH(see \cref{conj:qParitySAT}) is false.
    \end{enumerate}
\end{corollary}



\begin{corollary}\label{cor:gammacountQSETH}
Let $\gamma\in [\frac{1}{2^n},0.4999)$ and $p\in[1,\infty)\setminus 2\mathbb{Z}$. For each constant $\delta>0$, there is no bounded-error quantum algorithm that solves $\gamma$-VCP$_p$ in time
\begin{enumerate}
    \item $\mathcal{O}\left( \left( \frac{1}{\gamma}\sqrt{\frac{2^{n}-\hat{h}}{\hat{h}}}\right)^{1-\delta}\right)$, if $\gamma \hat{h} > 1$ where $\hat{h}$ is the number of the closest vectors,
    \item $\mathcal{O}(2^{(1-\delta)n})$, otherwise,
\end{enumerate} 
unless $\gamma$-\#QSETH (see \cref{thm:AppCountQSETH}, implied by \cref{thm:hardnessForSpecificEll}) is false. 
\end{corollary}

The following theorem also shows how to connect the approximate vector counting problem to the closest vector problem. The classical reduction from approximate VCP$_p$ to CVP$_p$ was already shown in~\cite[Theorem 3.5]{SteDiscreteGaussian16}, while we can easily give a quadratic saving for the number of calls to CVP$_p$ oracle. 
We include the proof at the end of this section for completeness.

\begin{theorem}\label{thm:VCPtoCVP}
Let $f(n)\geq 20$ be an efficiently computable function and $p\in [1,\infty)$. One can solve $f(n)^{-1}$-VCP$_p$ using $\mathcal{O}(f(n)^2)$ quantum queries to CVP$_p$.
\end{theorem}



To prove \cref{thm:VCPtoCVP}, we first introduce a gapped version of the lattice counting problem as follows.

\begin{defn}[Gap-VCP]
Let $\gamma\geq 0$ and $1\leq p \leq \infty$. The problem $\gamma$-approximate gap Vector Counting Problem $\gamma$-Gap-VCP$_p$ is a promise problem defined as follows: The input is a basis $\mathbf{B}\in \mathbb{R}^{d\times n}$ of a lattice $\cL(\mathbf{B})$, target vector $\vect{t}\in \mathbb{R}^{d}$, radius $r\in \mathbb{R}_+$, and $N\geq 1$. The goal of this problem is to output ``No'' if  $|(\cL-\vect{t})\cap r\cdot B^{d}_p|\leq N$ and ``Yes'' if  $ N > (1+\gamma)\cdot|(\cL-\vect{t})\cap r\cdot B^{d}_p|$.
\end{defn}

One can easily see that if we can solve $\gamma$-Gap-VCP$_p$, then by using $\poly(n)$ calls of it we can solve $\gamma$-VCP$_p$. As a result, it suffices to show how to reduce $\gamma$-Gap-VCP$_p$ to CVP$_p$ in the following proof.

Define $U_{Spar}:\ket{\mathbf{B},Q,\vect{z},\vect{c}}\ket{0}\ket{0}\rightarrow \ket{\mathbf{B},Q,\vect{z},\vect{c}}\ket{\mathbf{B}_{Q,\vect{z}}}\ket{\vect{w}_{\vect{z},\vect{c}}}$, where $\mathbf{B}_{Q,\vect{z}}, {\vect{w}_{\vect{z},\vect{c}}}$ are the output of $Spar(\mathbf{B},Q,\vect{z},\vect{c})$ (which is defined in \cref{thm:spar}). Since given CVP$_p$ oracle and a basis $\mathbf{B}$, the sparsification process can be efficiently done according to the construction, we can also implement the unitary $U_{Spar}$ efficiently. Then we are ready to show the proof of \cref{thm:VCPtoCVP}.

\begin{proof}[Proof of \cref{thm:VCPtoCVP}] 
Choose a prime $Q=\Theta(f(n)\cdot N)$ and let $O_{CVP}$ be the quantum CVP$_p$ oracle. Define $U_{Spar}:\ket{\mathbf{B},Q,\vect{z},\vect{c}}\ket{0}\ket{0}\rightarrow \ket{\mathbf{B},Q,\vect{z},\vect{c}}\ket{\mathbf{B}_{Q,\vect{z}}}\ket{\vect{w}_{\vect{z},\vect{c}}}$, where $\mathbf{B}_{Q,\vect{z}}, {\vect{w}_{\vect{z},\vect{c}}}$ are the output of $Spar(\mathbf{B},Q,\vect{z},\vect{c})$. Since given a CVP$_p$ oracle and a basis $\mathbf{B}$, the sparsification process (\Cref{thm:spar}) can be efficiently done according to the construction, we can also implement the unitary $U_{Spar}$ efficiently. 

First we prepare the superposition state $\frac{1}{Q^{n/2}}\sum\limits_{\vect{z},\vect{c}\in\intg^n_Q}\ket{\mathbf{B},Q,\vect{z},\vect{c}}\ket{0}\ket{\vect{t}}\ket{0}\ket{0}$, apply $U_{Spar}$ on the first six registers, and apply $O_{CVP_p}$ on the fifth, sixth, seventh registers, and then apply the subtraction unitary $U_{sub}:\ket{\vect{a}}\ket{\vect{b}}\ket{0}\rightarrow \ket{\vect{a}}\ket{\vect{b}}\ket{\|\vect{a}-\vect{b}\|_p}$ on the last three registers, then we get
$$
\frac{1}{Q^{n/2}}\sum\limits_{\vect{z},\vect{c}\in\intg^n_Q}\ket{\mathbf{B},Q,\vect{z},\vect{c}}\ket{\mathbf{B}_{Q,\vect{z}}}\ket{\vect{t}+\vect{w}_{\vect{z},\vect{c}}}\ket{CVP_p(\mathbf{B}_{Q,\vect{z}}, \vect{t}+\vect{w}_{\vect{z},\vect{c}})}\ket{\|\vect{t}+\vect{w}_{\vect{z},\vect{c}}-CVP_p(\mathbf{B}_{Q,\vect{z}}, \vect{t}+\vect{w}_{\vect{z},\vect{c}}))\|_p}.
$$
Let $r_{\vect{z},\vect{c}}=\|\vect{t}+\vect{w}_{\vect{z},\vect{c}}-CVP_p(\mathbf{B}_{Q,\vect{z}}, \vect{t}+\vect{w}_{\vect{z},\vect{c}}))\|_p$ and umcompute the first seven registers, then we get
$$
\frac{1}{Q^{n/2}}\sum\limits_{\vect{z},\vect{c}\in \intg_Q^n}\ket{r_{\vect{z},\vect{c}}}.
$$
Adding another ancilla $\ket{0}$ at the end of the above state, and then applying the $r$-threshold gate $$
U_r:\ket{R}\ket{0}\rightarrow \begin{cases}
    \ket{R}\ket{1} \text{ if } R\leq r \\
    \ket{R}\ket{0} \text{ otherwise, }
\end{cases}
$$ 
on it, we get
$$
\frac{1}{Q^{n/2}}\big(\sum\limits_{\substack{r_{\vect{z},\vect{c}}\leq r}}\ket{r_{\vect{z},\vect{c}}}\ket{1}+\sum\limits_{\substack{r_{\vect{z},\vect{c},}}> r}\ket{r_{\vect{z},\vect{c}}}\ket{0}\big):=\sqrt{a}\ket{\phi_1}\ket{1}+\sqrt{1-a}\ket{\phi_0}\ket{0},
$$
where $a=|\{(\vect{z},\vect{c}): r_{\vect{z},\vect{c}}\leq r\} |/Q^{n}=\Pr\limits_{\vect{z},\vect{c},\in \intg^n_Q}[r_{\vect{z},\vect{c}}\leq r]$.
Note that by \cref{thm:spar}, if $|\cL \cap (rB^n_p+\vect{t})|\leq N$, then 
$$
\Pr\limits_{\vect{z},\vect{c},\in \intg^n_Q}[r_{\vect{z},\vect{c}}\leq r] \leq \frac{N}{Q}+\frac{N}{Q^n},
$$
and if $|\cL \cap (rB^n_p+\vect{t})|\geq \gamma N$, then
$$
\Pr\limits_{z,c,\in \intg^n_Q}[r_{\vect{z},\vect{c}}\leq r] \geq \frac{\gamma N}{Q}-\frac{\gamma^2N^2}{Q^2}-\frac{\gamma^2N^2}{Q^{n-1}}.
$$
Observing that $\frac{\gamma N}{Q}-\frac{\gamma^2N^2}{Q^2}-\frac{\gamma^2N^2}{Q^{n-1}}-(\frac{N}{Q}+\frac{N}{Q^n})=\Theta(f(n)^{-1}N/Q)$, we know that to distinguish the above two cases, it suffices to learn $a=\Pr\limits_{\vect{z},\vect{c},\in \intg^n_Q}[r_{\vect{z},\vect{c}}\leq r]$ with additive error $\Theta(f(n)^{-1}N/Q)$. Therefore, by using \cref{thm:AmplitudeEstimation}, we can solve $\gamma$-Gap-VCP$_p$ with $\gamma=f(n)^{-1}$ using 
 $\mathcal{O}(f(n)Q/N)$ queries to $O_{CVP_p}$ and time. Because $Q=\Theta(f(n)N)$, we finish the proof.
\end{proof}

Note that the QSETH lower bound for $f(n)^{-1}$-VCP$_p$ depends on $f(n)$. Therefore if we can show a reduction from $f(n)^{-1}$-VCP$_p$ to CVP$_p$ using $f(n)^{c}$ for some constant $c<1$, then we will end up with a better QSETH lower bound for CVP$_p$.

\section{Hardness of Counting/Parity of OV, Hitting Set, and Set-Cover}\label{sec:OVandOtherResults}

In this section, we will discuss the consequences of \cref{thm:AppCountQSETH} and \cref{cor:countingksat} for some well-motivated optimization problems: Orthogonal Vectors, Hitting Set and Set Cover. Following are the definitions of Hitting Set and its variants.

\begin{defn}[Hitting Set]
\label{def:HS}
Let integers $n,m>0$. The Hitting Set problem is defined as follows: The input is a collection of sets $\Sigma=(S_1,\ldots,S_m)$, where $S_i\subset V$ and integer $t>0$, the goal is to output a subset $S'\subset V$ such that $|S'|\leq t$ and $\forall i\in [m], |S'\cap S_i|>0$. We call such $S'$ a \emph{hitting set} for $\Sigma$.
\end{defn}

\begin{defn}[Variants of Hitting Set]
    Let integers $n,m>0$ and $\gamma \in \left[ \frac{1}{2^n}, \frac{1}{2} \right) $. We define the following four variants of Hitting Set. The input for all of them is a collection $\Sigma=(S_1,\ldots,S_m)$, where $S_i\subset V$ and integer $t>0$.
    \begin{enumerate}
        \item In the Count Hitting Set problem, the goal is to output $d'$; 
        \item in the Parity Hitting Set problem, the goal is to output $d'\bmod 2$; 
        \item in the strict-Majority Hitting Set problem, the goal is to output $1$ if $d'> 2^{n-1}$, 
        otherwise output $0$;
        \item in the $\gamma$-approximation of count Hitting Set, the goal is to output an integer $d$ such that $(1-\gamma)d'<d<(1+\gamma)d'$;
    \end{enumerate}
    where $d'=|S'\subset V: \left|S'|\leq t,  \forall i\in [m], |S'\cap S_i|>0\right|$.
\end{defn}

In \cite{CDLMNOPSW16}, the authors showed a Parsimonious reduction between CNFSAT and Hitting Set. By Parsimonious reduction, we mean a transformation from a problem to another problem that preserves the number of solutions. 
 
\begin{theorem}[Theorem~3.4 in \cite{CDLMNOPSW16}]
    For each constant $\delta>0$, there exists a polynomial time Parsimonious reduction from CNFSAT on $n$ variables to Hitting Set on $n(1+\delta)$ size universal set.
\end{theorem}

By using the above theorem, we immediately get the following corollaries.
\begin{corollary}\label{cor:HitSetQSETH}
    For each constant $\delta>0$, there is no bounded-error quantum algorithm that solves
\begin{enumerate}
        \item Hitting Set in $\mathcal{O}(2^{\frac{n}{2}(1-\delta)})$ time,  unless \textsc{basic-QSETH} (see \cref{cor:basicQSETH}) is false.
        \item Count Hitting Set in $\mathcal{O}(2^{n(1-\delta)})$ time,  unless \#QSETH(see \cref{thm:countingQSETH}) is false.
        \item Parity Hitting Set in $\mathcal{O}(2^{n(1-\delta)})$ time,  unless $\oplus$QSETH(see \cref{thm:ParityQSETH})  is false.
        \item Strict-Majority Hitting Set in $\mathcal{O}(2^{n(1-\delta)})$ time,  unless Majority-QSETH(see \cref{thm:MajorityQSETH}) is false.
    \end{enumerate}
\end{corollary} 

\begin{corollary}\label{cor:gamma_HitSetQSETH}
 Let $\gamma\in \left[\frac{1}{2^n},0.4999\right) $. For each constant $\delta>0$, there is no bounded-error quantum that solves $\gamma$-multiplicative-factor approximation of count Hitting Set in time
\begin{enumerate}
    \item $\mathcal{O}\left(\frac{1}{\gamma} \sqrt{\frac{2^n-\hat{h}}{\hat{h}}}\right)^{1-\delta}$, if $\gamma \hat{h} > 1$ where $\hat{h}$ is the number of hitting sets,
    \item $\mathcal{O}(2^{(1-\delta)n})$, otherwise,
\end{enumerate} 
unless $\gamma$-\#QSETH (see \cref{thm:AppCountQSETH}, implied by \cref{thm:hardnessForSpecificEll}) is false. 
\end{corollary}

We also use our results from \cref{sec:LowerBoundsVariants} to show conditional lower bounds for problems in the complexity class $\textsf{P}$. More specifically, we study Orthogonal Vectors problem and its variants defined as follows.

\begin{defn}[Orthogonal Vectors (\OV)]
Let $d,n$ be natural numbers. The Orthogonal Vectors problem is defined as follows: The input is two lists A and B, each consisting of $n$ vectors from $\{0,1\}^d$. The goal is to find  vectors $a\in A, b\in B$ for which $\langle a,b \rangle=0$. We call such pair $(a,b)$ a pair of \emph{orthogonal vectors}.
\end{defn}

\begin{defn}[Variants of \OV]
   Let integers $d,n>0$, $\gamma \in \left[ \frac{1}{n^2}, 0.4999 \right) $, and the input is two lists $A$ and $B$ each consisting of $n$ vectors from $\{0,1\}^d$.
    \begin{enumerate}
        \item In the Count $\OV$ problem, the goal is to output $d'$; 
        \item in the Parity $\OV$ problem, the goal is to output $d'\bmod 2$; 
        \item in the strict-Majority $\OV$, the goal is to output $1$ if $d'>n^2/2$, 
        otherwise output $0$;
        \item in the $\gamma$-approximation of count $\OV$, the goal is to output an integer $d$ such that $(1-\gamma)d'<d<(1+\gamma)d'$; 
    \end{enumerate}
    where $d'=|(a,b): a\in A,b\in B, \langle a,b\rangle =0|$.
\end{defn}

Orthogonal Vectors is an important computational problem that lies in the complexity class $\textsf{P}$.  It turns out to be one of the central problems to show fine-grained hardness of problems in $\textsf{P}$~\cite{Vas15,abboud2018more}. Williams showed a reduction from CNF-SAT to $\OV$~\cite{williams2005new}. We observe that Williams's reduction is Parsimonious. Therefore we can also show quantum conditional lower bounds for counting versions of $\OV$ using our QSETH conjectures. In \cite{BPS21}, the authors showed $\mathcal{O}(n)$-hardness for $\OV$ under basic-QSETH assumption, and here we give $\mathcal{O}(n^2)$-hardness for counting versions of $\OV$, which might be useful for showing quantum conditional lower bounds for (variants of) other problems (like string problems or dynamic problems, see~\cite[Figure~1]{Vas15}).

\begin{corollary}\label{cor:OVQSETH}
For each constant $\delta>0$, there is no bounded-error quantum algorithm that solves
\begin{enumerate}
        \item Count $\OV$ in $\mathcal{O}(n^{2-\delta})$ time,  unless \#QSETH (see \cref{thm:countingQSETH}) is false;
        \item  Parity $\OV$ in $\mathcal{O}(n^{2-\delta})$ time, unless $\oplus$QSETH (see \cref{thm:ParityQSETH}) is false;
        \item Majority $\OV$ in $\mathcal{O}(n^{2-\delta})$ time, unless Majority-QSETH (see \cref{thm:MajorityQSETH}) is false.
    \end{enumerate}
\end{corollary} 

\begin{corollary}\label{cor:gamma_OVQSETH}
 Let $\gamma\in \left[\frac{1}{2^n},0.4999\right) $. For each constant $\delta>0$, there is no bounded-error quantum algorithm that solves $\gamma$-multiplicative-factor approximation of count $\OV$ in time
\begin{enumerate}
    \item $\mathcal{O}\left(\frac{1}{\gamma} \sqrt{\frac{n^2-\hat{h}}{\hat{h}}}\right)^{1-\delta}$, if $\gamma \hat{h} > 1$ where $\hat{h}$ is the number of pairs of orthogonal vectors,
    \item $\mathcal{O}(n^{(2-\delta)})$, otherwise,
\end{enumerate} 
unless $\gamma$-\#QSETH (see \cref{thm:AppCountQSETH}, implied by \cref{thm:hardnessForSpecificEll}) is false. 
\end{corollary}

We can also give a quantum conditional lower bound for the parity Set-Cover problem defined as follows.

\begin{defn}[parity Set-Cover]
    For any integers $n,m>0$, the parity Set-Cover problem is defined as follows: The input is a collection $\Sigma=(S_1,\ldots,S_m)$, where $S_i\subset V$ and integer $t>0$, the goal is to output $\left|\{\mathcal{F}\subset \Sigma: \bigcup_{S\in \mathcal{F}}S=V, |\mathcal{F}|\leq t \}\right|\mod 2$.
\end{defn}

In \cite{CDLMNOPSW16}, the authors showed an efficient reduction from parity Hitting Set to parity Set-Cover. Using the third item of \cref{cor:HitSetQSETH} we get the following corollary:

\begin{corollary}\label{cor:paritySetCoverQSETH}
    For each constant $\delta >0$, there is no bounded-error quantum algorithm that solves parity Set-Cover in $2^{n(1-\delta)}$ time, unless $\oplus$QSETH (see \cref{thm:ParityQSETH}) is false.
\end{corollary}

\section{Discussion and open questions}
\label{sec:Conclusion}
We believe that this paper opens up the possibility of concluding quantum time lower bounds for many other problems, both for other variants of \CNFSAT{} and also for problems that are not immediately related to \CNFSAT{}. While this is a natural broad future direction to explore, we also mention the following few directions for future work which are more contextual to this paper.
\begin{itemize}
    \item One of the motivations to use \acQSETH{} in this paper is so that we can `tie' certain conjectures, that would have otherwise been standalone conjectures, to one main conjecture. But in the process, we conjecture compression obliviousness of several properties. It would be nice if we could also have an `umbrella' conjecture that allows one to establish compression obliviousness of several properties. For e.g., it would be nice if we could show that compression obliviousness of a natural property like \propertyCount{} or \propertyParity{} implies compression obliviousness of say \propertyAddErrorCount{}.
     
    \item It will be interesting to see if it is possible to use the QSETH framework (or the \acQSETH{} conjecture) to give a single exponential lower bound for \#CVP in Euclidean norm. 

    \item Using Boolean function Fourier analysis, we were able to show that the existence of (quantum-secure) PRFs imply that majority and parity are compression oblivious, whenever the input is given by a formula or circuit. This proof technique could plausibly be extended to larger sets of functions that have a similar structure, e.g., a natural candidate would be to show an equivalent statement for symmetric functions with non-negligible mass on high-degree Fourier coefficients.
    
    Additionally, extending this result to majority / parity for \acQSETH{}, i.e.\ CNF or DNF input, would be another step towards grounding the (necessary) assumption that such properties are compression oblivious.
\end{itemize}

\bibliographystyle{alpha}
\bibliography{Lattice.bib}

\appendix

\end{document}

%% file: preamble.tex
\usepackage{fullpage}
\usepackage{graphicx}
\usepackage{upref}
\usepackage{enumerate}
\usepackage{latexsym}
\usepackage{pdfpages}
\usepackage{tikz}
\usetikzlibrary{shapes, arrows}
\usepackage[bookmarks]{hyperref}
\usepackage{sansmath}

\usepackage{algorithm}
\usepackage[noend]{algpseudocode}

\usepackage{pdflscape}

\usepackage{color,graphics}
\usepackage{comment} 
\usepackage{caption}
\usepackage{subcaption}
\usepackage{wrapfig}
\usepackage{braket}
\usepackage{amssymb}
\usepackage{amsmath}
\usepackage{amsthm}
\usepackage{mathrsfs}
\usepackage{mathtools}
\usepackage{commath}
\usepackage{thmtools,thm-restate}
\usepackage{tabularx, makecell, multirow}
\usepackage{tikz}
\usetikzlibrary{decorations.pathreplacing,decorations.pathmorphing}
\usepackage{boxedminipage}
\usetikzlibrary{patterns}
\usetikzlibrary{arrows.meta, positioning}

\DeclareMathOperator{\cL}{\mathcal{L}}

\DeclareMathOperator{\basis}{\mathbf{B}}

\usepackage{amsfonts}
\usepackage{varioref}
\usepackage[ansinew]{inputenc}
\usepackage[many]{tcolorbox}

\usepackage{xcolor}
\hypersetup{
	colorlinks,
	linkcolor={red!75!black},
	citecolor={blue!75!black},
	urlcolor={blue!75!black}
}
\usepackage[hyperpageref]{backref}

\newtheorem{theorem}{Theorem}[section]
\newtheorem{lemma}[theorem]{Lemma}
\newtheorem{corollary}[theorem]{Corollary}
\newtheorem{conjecture}[theorem]{Conjecture}
\newtheorem{claim}[theorem]{Claim}

\newtheorem{defn}[theorem]{Definition}

\theoremstyle{definition}

\AtBeginDocument{}

\DeclareMathOperator{\real}{\mathbb{R}}

\newcommand{\intg}{\mathbb{Z}}

\newcommand{\inProd}[2]{\langle{#1},{#2}\rangle}

\newcommand{\ceil}[1]{\lceil{#1}\rceil}

\newcommand{\poly}{\mathrm{poly}}

\newcommand{\appCountingCNFSAT}{$\gamma$-\#\textsc{CNFSAT}}
\newcommand{\addErrorCountingCNFSAT}{$\Delta$-{\textsc{add}}-\#\textsc{CNFSAT}}
\newcommand{\kSAT}{$k$\textsc{-SAT}}

\newcommand{\CountingCNFSAT}{\#\textsc{CNFSAT}}
\newcommand{\ParityCNFSAT}{$\oplus$\textsc{CNFSAT}}
\newcommand{\qParityCNFSAT}{$\#_q$\textsc{CNFSAT}}
\newcommand{\majorityCNFSAT}{\textsc{maj-CNFSAT}}
\newcommand{\stMajorityCNFSAT}{\text{st}\textsc{-maj-CNFSAT}}
\newcommand{\CNFSAT}{\textsc{CNFSAT}}
\newcommand{\acQSETH}{$\textsc{AC}^0_2$\textsc{-QSETH}}
\newcommand{\BasicQSETH}{\textsc{Basic-QSETH}}
\newcommand{\QSETH}{\textsc{QSETH}}

\newcommand{\reduceWidth}{\ensuremath{\text{ReduceWidth}_{k}}}
\newcommand{\OV}{\textsf{OV}}

\newcommand{\sol}{\text{sol}}

\newcommand{\CNF}{\textsc{CNF}}
\newcommand{\DNF}{\textsc{DNF}}
\newcommand{\AC}{\mathsf{AC}_{2}^{0}}

\newcommand{\ACp}{\mathsf{AC}_{2,p}^{0}}
\newcommand{\ACarg}[2]{\mathsf{AC}_{{#1},{#2}}^{0}}

\newcommand{\propertyP}{\textsc{P}}

\newcommand{\propertyOR}{\textsc{OR}}

\newcommand{\propertyCount}{\textsc{count}}
\newcommand{\propertyParity}{\textsc{parity}}
\newcommand{\propertyMajority}{\textsc{majority}}

\newcommand{\propertyStrictMajority}{\textit{st}\textsc{-majority}}

\newcommand{\propertyAddErrorCount}{\ensuremath{\Delta \textsc{-add-count}}}

\newcommand{\vect}[1]{\boldsymbol{#1}}

\renewcommand{\epsilon}{\varepsilon}
\renewcommand{\vec}[1]{\vect{#1}}

\newcommand{\R}{\mathbb{R}}

\newcommand{\rnote}[1]{}

\newcommand{\Q}{\mathsf{Q}}

\usepackage[capitalize,nameinlink,noabbrev]{cleveref}
\crefname{conjecture}{Conjecture}{Conjectures}
\crefname{defn}{Definition}{definition}
\crefname{claim}{Claim}{Claims}
\creflabelformat{equation}{#2\textup{#1}#3}

\newenvironment{claimproof}[1]{\par\noindent\underline{Proof:}\space#1}{\hfill $\blacksquare$}

%% file: TableOfCNFResults.tex
\begin{table}[!ht]
\centering
\begin{tabular}{|l|l|l|l|}
\hline
                 Problem & Variants & Quantum lower bound & Reference \\ \hline 
&& \\[-1.25em]
\multirowcell{8}{\\ \\ \\ \textsc{\CNFSAT{}} } 
                  & Vanilla & $2^{\frac{n}{2}-o(n)}$  &\cref{cor:basicQSETH} \\ \cline{2-4}&&\\[-1.2em] 
                  & Parity & $2^{n-o(n)}$  &\cref{thm:ParityQSETH} \\ \cline{2-4}&&\\[-1.2em] 
                  & Majority & $2^{n-o(n)}$  &\cref{thm:MajorityQSETH} \\ \cline{2-4}&&\\[-1.2em] 
                  & Strict Majority & $2^{n-o(n)}$  &\cref{thm:MajorityQSETH} \\ \cline{2-4}&&\\[-1.2em] 
                  & Count & $2^{n-o(n)}$  &\cref{thm:countingQSETH} \\ \cline{2-4}&&\\[-1.2em] 
                  & $\text{Count}_q$ & $2^{n-o(n)}$  &\cref{conj:qParitySAT} \\ \cline{2-4}&&\\[-1.2em] 
                  & $\Delta$-Additive error & $\left( \sqrt{\frac{2^n}{\Delta}} + \frac{\sqrt{\hat{h} (2^n -\hat{h})}}{\Delta}\right)^{1-o(1)}$ &\cref{thm:AddErrorQSETH}\\ \cline{2-4}&&
\\[-1.2em] & $\gamma$-Multiplicative factor & $\left( \frac{1}{\gamma}\sqrt{\frac{2^{n}-\hat{h}}{\hat{h}}}\right)^{1-o(1)}$ &\cref{thm:AppCountQSETH} \\ \cline{2-4} \hline

&&\\[-1.25em]
\multirowcell{4}{\\ \\ $k$-\textsc{SAT} \\ $k=\Theta(\log(n))$} & Vanilla & $2^{\frac{n}{2}-o(n)}$ & \cref{sec:lattice}, \cite{BasicQSETH-Aaronson-2020} \\ \cline{2-4}&& \\[-1.2em]
                  & Parity & $2^{n-o(n)}$  &\cref{cor:countingksat} \\ \cline{2-4}&&\\[-1.2em]
                  & Count & $2^{n-o(n)}$ &\cref{cor:countingksat} \\ \cline{2-4}\\[-1.2em]
                  & $\text{Count}_q$ & $2^{n-o(n)}$ &\cref{cor:countingksat} \\ \cline{2-4}\\[-1.2em]
                  & $\gamma$-Multiplicative factor & $\left(\frac{1}{\gamma} \sqrt{\frac{2^n-\hat{h}}{\hat{h}}}\right)^{1-o(1)}$ &\cref{cor:MultErrorkSAT}\\ \hline

\end{tabular}
\caption{Overview of conditional lower bounds for variants of \CNFSAT{} and \kSAT{}. The variable $\hat{h}$ in the above table is an arbitrary natural number satisfying $\gamma\hat{h}\geq 1$. Our results hold for the multiplicative factor $\gamma \in \left[\frac{1}{2^{n}},0.4999\right)$ and the additive error $\Delta \in [1,2^n)$.}
\label{table:SummaryOfCNFResults}
\end{table}

%% file: TableOfResults.tex
\begin{table}[!ht]
\centering
\begin{tabular}{|l|l|l|l|}
\hline
                 Problem & Variants & Quantum lower bound & Reference \\ \hline 
&& \\[-1.25em]
\multirowcell{4}{\\ \\ \textsc{Strong Simulation} } & Exact (with $n$ bits precision) & $2^{n-o(n)}$  &\cref{thm:QHardnessExactStrongSimulation} \\ \cline{2-4}&&\\[-1.2em] 
& Exact (with $2$ bits precision) & $2^{n-o(n)}$  &\cref{thm:StrongSimulationExactTwoBits} \\ \cline{2-4}&&\\[-1.2em] 
                  & $\Delta$-Additive error & $\left( \sqrt{\frac{1}{\Delta'}} + \frac{\sqrt{\hat{h} (2^n -\hat{h})}}{2^n\Delta'}\right)^{1-o(1)}$ &\cref{cor:DeltaErrorQSim}\\ \cline{2-4}&&
\\[-1.2em] & $\gamma$-Multiplicative factor & $\left( \frac{1}{\gamma}\sqrt{\frac{2^{n}-\hat{h}}{\hat{h}}}\right)^{1-o(1)}$ &\cref{thm:gammaQsim} \\ \cline{2-4} \hline  

&&\\[-1.25em]
\multirowcell{1}{\textsc{CVP$_p$ for $p\notin 2\mathbb{Z}$}}
                  &  & $2^{\frac{n}{2}-o(n)}$ & \cref{sec:lattice} \\ \hline

&&\\[-1.25em]
\multirowcell{3}{\\ \textsc{Lattice Counting} \\ \text{(for non-even norm)}} & Vanilla &  $2^{n-o(n)}$ &\cref{cor:latticecountQSETH} \\ \cline{2-4}&& \\[-1.2em]
                  & $\gamma$-Multiplicative factor & $\left( \frac{1}{\gamma}\sqrt{\frac{2^{n}-\hat{h}}{\hat{h}}}\right)^{1-o(1)}$ &\cref{cor:gammacountQSETH}
                 \\  \cline{2-4}
                  & $q$-count & $2^{n-o(n)}$ &\cref{cor:qcountlatticeQSETH}
                 \\ \hline

&&\\[-1.25em]
\multirowcell{4}{\\ \textsc{Orthogonal Vectors} \\} & Vanilla & $n^{1-o(1)}$ &\cite{BasicQSETH-Aaronson-2020, BPS21} \\ \cline{2-4}&& \\[-1.2em]
                  & Parity & $n^{2-o(1)}$  &\cref{cor:OVQSETH} \\ \cline{2-4}&&\\[-1.2em]
                  & Count & $n^{2-o(1)}$ &\cref{cor:OVQSETH} \\ \cline{2-4}\\[-1.2em]
                  & $\gamma$-Multiplicative factor & $\left(\frac{1}{\gamma} \sqrt{\frac{n^2-\hat{h}}{\hat{h}}}\right)^{1-o(1)}$ &\cref{cor:OVQSETH}\\ \hline

&&\\[-1.25em]
\multirowcell{4}{\\ \textsc{Hitting Set} \\} & Vanilla & $2^{\frac{n}{2}-o(n)}$  &\Cref{cor:HitSetQSETH} \\ \cline{2-4}&& \\[-1.2em]
                  & Parity & $2^{n-o(n)}$ &\cref{cor:HitSetQSETH}\\ \cline{2-4}&&\\[-1.2em]
                  & Count & $2^{n-o(n)}$ &\cref{cor:HitSetQSETH}\\ \cline{2-4}\\[-1.2em]
                  & $\gamma$-Multiplicative factor & $\left(\frac{1}{\gamma} \sqrt{\frac{2^n-\hat{h}}{\hat{h}}}\right)^{1-o(1)}$ &\cref{cor:HitSetQSETH} \\ \hline

&&\\[-1.25em]
\multirowcell{1}{\textsc{$\oplus$Set Cover}}
                  &  & $2^{n-o(n)}$ &\cref{cor:paritySetCoverQSETH} \\ \hline
\end{tabular}
\caption{Overview of lower bounds based on \acQSETH{}. The variable $\hat{h}$ in the above table is an arbitrary natural number satisfying $\gamma\hat{h}\geq 1$. Our results hold for the multiplicative factor $\gamma \in \left[\frac{1}{2^{n}},0.4999\right)$ and the additive error $\Delta' \in [\frac{1}{2^n},1)$.}
\label{table:SummaryOfResults}
\end{table}

%% file: flowchart.tex
\scalebox{0.7}{
\tikzstyle{terminator} = [rectangle, draw, text centered, rounded corners, minimum height=2em]
\tikzstyle{process} = [rectangle, draw, text centered, minimum height=2em]

\tikzstyle{decision} = [diamond, draw, text width=3.5cm, align=center, minimum height=2em, minimum width=2em]
\tikzstyle{data}=[trapezium, draw, text centered, trapezium left angle=60, trapezium right angle=120, minimum height=2em]
\tikzstyle{connector} = [draw, -latex']
\begin{tikzpicture}
\node [terminator, fill=blue!20] at (0,0) (start) {\textbf{Start}};
\node [data, fill=blue!20] at (0,-1.25) (data) {Goal: for a given property $\textsc{P}$ show the quantum time complexity of $\textsc{P-}$\CNFSAT{}};
\node [process, fill=blue!20] at (0,-2.5) (AC-QSETH) {Invoke \acQSETH{} (\Cref{conj:ACqseth})};
\node [decision, fill=blue!20] at (0,-6.25) (decisionEasyCO) {Easy to show $\textsc{P} \in \mathcal{CO}(\textsc{AC}_2^0)$?  See Theorem~9 in \cite{BPS21}.};
\node [process, fill=blue!20] at (0,-10) (conjectureCO) {Conjecture compression obliviousness of $\textsc{P}$};
\node [process, fill=blue!20] at (4.25,-6.25) (ProveIt) {Prove it};
\node [decision, fill=blue!20] at (0,-13.5) (decisionCO) {Is $\textsc{P} \in \mathcal{CO}(\textsc{AC}_2^0)$ (provably or in conjecture)?};
\node [process, fill=red!20] at (6.25,-13.5) (NotApplicable) {Cannot use \acQSETH{} for \textsc{P}};
\node [process, fill=blue!20] at (0,-17) (Q_P) {Compute the bounded-error quantum query complexity of $\textsc{P}$, i.e. $Q(\textsc{P})$};
\node [process, fill=blue!20] at (0,-18.25) (acQSETHlowerBoundP) {Quantum time lower bound for computing $\textsc{P}$ on \CNF{} and \DNF{} is $\Omega(Q(\textsc{P})^{1-\delta})$ for all $\delta>0$};
\node [decision, fill=blue!20] at (0,-22.25) (CNFvsDNFforP) {Are \CNF{}s at least harder in comparison to \DNF{}s for computing $\textsc{P}$};
\node [process, fill=red!20] at (8,-22.25) (NoComments) {Cannot comment on hardness of $\textsc{P-}\CNFSAT{}$};
\node [process, fill=blue!20] at (0,-26.25) (P-CNFSAT) {Complexity of $\textsc{P-}\CNFSAT{}$ is $\Omega(Q(\textsc{P})^{1-\delta})$ for all $\delta>0$, under \acQSETH{}};
\node [terminator, fill=blue!20] at (0,-27.5) (end) {\textbf{End}};

\node[draw=none] at (3, -5.75) (yes) {Yes};
\node[draw=none] at (0.5, -9.25) (no) {No};

\node[draw=none] at (3, -13.25) (no) {No};
\node[draw=none] at (0.5, -16.25) (yes) {Yes};

\node[draw=none] at (3.5, -22) (no) {No};
\node[draw=none] at (0.5, -25.5) (yes) {Yes};

\node[draw=none] at (4.25, -10.7) (invisibleNode) {};

\path [connector] (start) -- (data);
\path [connector] (data) -- (AC-QSETH);
\path [connector] (AC-QSETH) -- (decisionEasyCO);
\path [connector] (decisionEasyCO) -- (conjectureCO);
\path [connector] (conjectureCO) -- (decisionCO);
\path [connector] (decisionCO) -- (Q_P);
\path [connector] (Q_P) -- (acQSETHlowerBoundP);
\path [connector] (acQSETHlowerBoundP) -- (CNFvsDNFforP);
\path [connector] (CNFvsDNFforP) -- (P-CNFSAT);
\path [connector] (P-CNFSAT) -- (end);
\path [connector] (decisionEasyCO) -- (ProveIt);
\path [connector] (decisionCO) -- (NotApplicable);
\path [connector] (ProveIt) -- (invisibleNode);
\path [connector] (CNFvsDNFforP) -- (NoComments);

\draw [->] (4.25,-10.55) -- (0,-10.55);

\end{tikzpicture}
}

%% file: CircuitStrongSimulationCircuit.tex
\scalebox{1.25}{
\begin{tikzpicture}[scale=1.000000,x=1pt,y=1pt]


\draw[decorate,decoration={brace,amplitude = 4.000000pt}, thick] (-10.000000,-25.750000) -- (-10.000000,45.750000);
\draw[color=black] (-12.000000,10.000000) node[left] {$n$};

\draw[color=white] (115.000000,0.000000) node[left] {$n$};

\draw[color=black] (0.000000,40.000000) -- (90.000000,40.000000);
\node[rectangle,
    draw = black,
    text = olive,
    fill = white,
    minimum width = 0.25cm, 
    minimum height = 0.5cm] (r) at (17.5,40) {$\textcolor{black} H$};

\node[rectangle,
    draw = black,
    text = olive,
    fill = white,
    minimum width = 0.25cm, 
    minimum height = 0.5cm] (r) at (72.5,40) {$\textcolor{black} H$};

\draw[color=black] (0.000000,20.000000) -- (90.000000,20.000000);
\node[rectangle,
    draw = black,
    text = olive,
    fill = white,
    minimum width = 0.25cm, 
    minimum height = 0.5cm] (r) at (17.5,20) {$\textcolor{black} H$};

\node[rectangle,
    draw = black,
    text = olive,
    fill = white,
    minimum width = 0.25cm, 
    minimum height = 0.5cm] (r) at (72.5,20) {$\textcolor{black} H$};

\draw[color=black] (0.000000,-20.000000) -- (90.000000,-20.000000);
\node[rectangle,
    draw = black,
    text = olive,
    fill = white,
    minimum width = 0.25cm, 
    minimum height = 0.5cm] (r) at (17.5,-20) {$\textcolor{black} H$};

\node[rectangle,
    draw = black,
    text = olive,
    fill = white,
    minimum width = 0.25cm, 
    minimum height = 0.5cm] (r) at (72.5,-20) {$\textcolor{black} H$};

\draw[color=black] (0.000000,-40.000000) -- (90.000000,-40.000000);

\draw[color=black] (0.000000,-60.000000) -- (90.000000,-60.000000);

\draw[color=black] (0.000000,-80.000000) -- (90.000000,-80.000000);
\node[rectangle,
    draw = black,
    text = olive,
    fill = white,
    minimum width = 0.25cm, 
    minimum height = 0.5cm] (r) at (72.5,-80) {$\textcolor{black} X$};

\node[rectangle,
    draw = black,
    text = olive,
    fill = white,
    minimum width = 0.25cm, 
    minimum height = 4.75cm] (r) at (45,-20) {$\textcolor{black}{C''_{\phi}}$};

\end{tikzpicture}
}